\documentclass[journal]{new-aiaa} 
\usepackage[utf8]{inputenc}
\pdfoutput=1

\usepackage[version=4]{mhchem}
\usepackage{siunitx}
\usepackage{longtable,tabularx}
\usepackage{algorithm}
\usepackage[noend]{algpseudocode}

\usepackage{amssymb}
\usepackage{amsthm}
\usepackage{amsfonts}
\usepackage{nccmath}

\newtheorem{theorem}{Theorem}

\newtheorem{lemma}[theorem]{Lemma}
\newtheorem{definition}[theorem]{Definition}
\newtheorem{remark}[theorem]{Remark}
\newtheorem{proposition}[theorem]{Proposition}

\newtheorem{problem}[theorem]{Problem}
\newtheorem{assumption}[theorem]{Assumption}

\usepackage{acronym}
\newacro{DGR}{Data-Guided Regulation}
\newacro{DG-RAN}{Data-Guided Regulation for Adaptive Nonlinear Control}
\newacro{F-DGR}{Fast Data-Guided Regulation}
\newacro{SVD}{Singular Value Decomposition}
\newacro{PBH}{Popov-Belevitch-Hautus}
\newacro{LQR}{Linear Quadratic Regulator}
\newacro{PI}{Policy Iteration}
\newacro{RLS}{Recursive Least Squares}
\newacro{sysID}{system identification}
\newacro{LMI}{Linear Matrix Inequalities}
\newacro{LTI}{Linear Time-Invariant}
\newacro{MPC}{Model Predictive Control}
\newacro{e-ISS}{exponentially input-to-state stable}
\newacro{uccm}{Unmatched Control Contraction Metric}
\newacro{uclf}{Unmatched Control Lyapunov function}
\newacro{RL}{Reinforcement Learning}
\newacro{clf}{control Lyapunov function}
\newacro{DMI}{Differential Matrix Inequality}

\usepackage[utf8]{inputenc}
\usepackage{color}
\usepackage{graphicx}
\usepackage{mathtools}
\usepackage{subfig}
\usepackage{caption}
\usepackage{tabularx}
\usepackage{amsmath}

\usepackage{amsthm}
\usepackage{amssymb}
\usepackage{amsfonts}
\usepackage{amstext}
\usepackage{xcolor}

\usepackage{tikz}
\usepackage{tikz-3dplot} 

\usepackage{mathtools}
\usepackage{float}
\usepackage{balance}
\usepackage[noabbrev]{cleveref}
\crefformat{equation}{(#1)}
\usepackage{bbm}
\usepackage{bm}
\usepackage{resizegather}
\usepackage{algorithm}
\usepackage[noend]{algpseudocode}

\newcommand\x{{\bm x}}
\def\u{{\bm u}}
\def\w{{\bm w}}

\def\z{{\bm z}}
\def\v{{\bm v}}
\def\tK{{\mathbf{\tilde{K}}}}

\newcommand\ogamma{{\bm \gamma}}

\newcommand{\RomanNumeralCaps}[1]
    {\MakeUppercase{\romannumeral #1}}

\usepackage{newtxtext,newtxmath}

\title{Data-Guided Regulator for Adaptive Nonlinear Control\footnote{The research of the first author has been partially supported by the UW+Amazon Science Hub Fellowship. This research has also been supported by NSF grant ECCS-2149470.}}
\author{Niyousha Rahimi\footnote{Ph.D. Candidate, AIAA Student Member, nrahimi@uw.edu} and   Mehran Mesbahi\footnote{Professor, AIAA Associate Fellow, mesbahi@uw.edu}}
\affil{Dept. of Aeronautics \& Astronautics, University of Washington, Seattle, WA 98195, USA}

\begin{document}

\maketitle

\begin{abstract}
    This paper addresses the problem of designing a data-driven feedback controller for complex nonlinear dynamical systems in the presence of  time-varying disturbances with unknown dynamics. Such disturbances are modeled as the ``unknown'' part of the system dynamics. The goal is to achieve finite-time regulation of system states through direct policy updates while also generating informative data that can subsequently be used for data-driven stabilization or system identification. 
    First, we expand upon the notion of ``regularizability'' and characterize this system characteristic for a linear time-varying representation of the nonlinear system with locally-bounded higher-order terms. ``Rapid-regularizability'' then gauges the extent by which a system can be regulated in finite time, in contrast to its asymptotic behavior.
    
    We then propose the \acf{DG-RAN} algorithm, an online iterative synthesis procedure that utilizes discrete time-series data from a single trajectory for regulating system states and identifying disturbance dynamics. The effectiveness of our approach is demonstrated on a 6-DOF power descent guidance problem in the presence of adverse environmental disturbances.
    
\end{abstract}

\section{Nomenclature}

{\renewcommand\arraystretch{1.0}
\noindent\begin{longtable*}{@{}l @{\quad=\quad} l@{}}
\multicolumn{2}{@{}l}{Notation}\\
$\mathbb{R}$, $\mathbb{C}$ & sets of real and complex numbers \\
$\mathbb{R}^n$, $\mathbb{R}^{n \times m}$ & $n$-dimensional Euclidean space and real $n\times m$ matrix \\
$\mathbbm{1}$, $0_{n\times m}$ & $n \times 1$ vector of all ones and $n\times n$ matrix of zeros \\ 
 $\mathrm{I}$, $\mathrm{I}_n$ & $n\times n$ identity matrix \\
$\lfloor a \rfloor$ & The floor function of a real number $a\in \mathbb{R}$\\
$L\succ 0$, $L\succeq 0$ & real positive-definite (PD) and real positive-semidefinite (PSD) matrices \\
$M^\top$ & transpose of a real matrix $M\in \mathbb{R}^{n \times m}$  \\
$\mathcal{R}(M) \subseteq \mathbb{R}^n$ & range of a real matrix $M \in \mathbb{R}^{n \times m}$ \\
$\mathcal{N}(M)\subseteq \mathbb{R}^m$ & null-space of a real matrix $M \in \mathbb{R}^{n \times m}$\\
$\mathrm{rank} (M)$ & dimension of $\mathcal{R}(M)$ \\
$\mathrm{span}\{.\}$ & the span of a set of vectors over the complex field\\
$M= {U} \Sigma {V}^\intercal$ & the \ac{SVD} of a matrix $M \in \mathbb{R}^{n\times m}$ \\
$U\in \mathbb{R}^{n\times n}$ & unitary matrix consists of left  ``singular" vectors of $M$\\
$V\in \mathbb{R}^{m\times m}$  & unitary matrix consists of right ``singular" vectors of $M$\\
$M = U_r \Sigma_r V_r^\intercal$ & the thin \ac{SVD} of $M$, where $r = \mathrm{rank}(M)$ and $\Sigma_r \in \mathbb{R}^{r \times r}$ is now nonsingular.\\
$U_r\in \mathbb{R}^{n\times r}$  & reduced order matrix obtained by truncating the factored $U\in \mathbb{R}^{n\times n}$ in the \ac{SVD} to the first $r$ columns\\
$V_r\in \mathbb{R}^{r\times m}$ & reduced order matrix obtained by truncating the factored $V\in \mathbb{R}^{m\times m}$ in the \ac{SVD} to the first $r$ columns\\
$M^{\dagger}={V}\Sigma^{\dagger}{U}^\intercal$ & Moore-Penrose generalized inverse ---\textit{pseudoinverse} for short--- of $M$ \\
$\rho(A)$ & the spectral radius of a square matrix $A \in \mathbb{R}^{n \times n}$, i.e., maximum modulus of eigenvalues of $A$\\
$\rho(A) \leq 0$ & denotes stability of the square matrix $A \in \mathbb{R}^{n\times n}$\\ $~\Pi_{S}(\bm v)$ & orthogonal projection of a vector $\bm v$ on a linear subspace $S$ \footnote{We will be working with finite dimensional vector spaces and as such all subspaces are closed.}   \\
$~\Pi_{S} = U U^\intercal$ & when the columns of matrix $U\in \mathbb{R}^{n \times k}$ form an orthonormal basis for the linear subspace $S$ \\
$\| \bm v \|$ & {Euclidean norm} of a vector $\bm v \in\mathbb{R}^n$ \\
$\|w(\cdot)\|_{\infty}$ &  $\sup\limits_{t\geq 0} \|w(t)\|$, infinity norm of a continuous signal $w(t)$\\
$\|w(\cdot)\|_{\Delta t}$ & $\sup\limits_{t_f \geq t\geq t_0} \|w(t)\|$, infinity norm of a continuous signal $w(t)$ for $t\in[t_0,t_f]$
\end{longtable*}}

\section{INTRODUCTION}
\label{sec:intro}
A critical aspect of autonomous operations in safety-critical scenarios is learning from available data for quick adaptation to new environments while maintaining safety. Examples include aircraft emergency landing scenarios in adverse weather conditions and agile quadrotor flights through low clearance gates in the presence of dynamic and strong wind conditions \cite{o2022neural}. 
From a system theoretic perspective, this system feature maps to having the autonomous agent  handle parametric model uncertainties and disturbances with control-theoretic guarantees such as stability and tracking error convergence, common in adaptive control settings~\cite{slotine1991applied,aastrom2013adaptive}. 
A rich body of literature has analyzed classical adaptive control algorithms' stability and convergence properties for continuous-time dynamical systems. Such studies include the use of PI (proportional integral) controllers \cite{accikmecse2003robust} for a class of linear time-varying systems to guarantee (\RomanNumeralCaps{1}) \textbf{infinite-time} convergence of the tracking error to zero, i.e., the difference between actual and nominal states $e(t) = x(t) - \bar{x}(t)$, for any constant exogenous disturbance (denoted by $w$), (\RomanNumeralCaps{2}) \textbf{infinite-time} convergence of the tracking error $e(t)$ to a bound which is proportional to the bound on the magnitude of the rate of the exogenous signal $\dot{w}(t)$. Furthermore, initial work in adaptive control of nonlinear systems used Lyapunov-like stability arguments and the certainty equivalence principle to construct stabilizing adaptive feedback
policies for feedback linearized systems with matched uncertainties \cite{slotine1986adaptive, slotine1987adaptive, taylor1988adaptive, sastry1989adaptive, kanellakopoulos1991extended}, i.e., those that can be directly canceled through control. However, the certainty equivalence principle cannot be employed when the model uncertainties are outside
the span of the control input, i.e., unmatched uncertainty. Departure from certainty equivalence significantly complicates the design process for a nonlinear system, as the controller must either anticipate or be robust to transients in the parameter estimates.
Only recently, the work~\cite{talebi2021regularizability} has studied the input “effectiveness” as it relates to \textbf{finite-time} regulation for a class of partially unknown linear systems, where the unknown part of the system can be modeled as unmatched uncertainty.

Many recent studies aim to integrate ideas from learning, optimization, and control theory to design and analyze adaptive controllers using learning-theoretic metrics. 
Precisely, due to recent successes of \ac{RL} in the control of physical systems, \cite{yang2020reinforcement,hwangbo2019learning,williams2017information}, there has been a surge of research activities in online \ac{RL} algorithms for continuous control. In contrast to the classical setting of adaptive nonlinear control, online \ac{RL} algorithms operate in discrete time and often come with {\em finite-time}~\cite{shi2021meta,boffi2021regret,dean2020sample,lale2020logarithmic,kakade2020information} or dynamic regret bounds \cite{li2019online,yu2020power,lin2021perturbation}. These bounds provide a quantitative rate at which the control performance of the online algorithm approaches the performance of an oracle equipped with hindsight knowledge of the uncertainty. In such scenarios, the controller does not necessarily guarantee the stability or robustness of the synthesized system; rather, the goal is to match the performance of the oracle for a well-behaved dynamical model.

When dealing with uncertain dynamical systems, most of the aforementioned learning-based studies focus on two restricting scenarios. The first category of work focuses on a finite-horizon episodic setting \cite{kakade2020information, lale2020logarithmic,dean2020sample}, which might not be directly applicable to online safety-critical control such as online flight control. The second category of work either assume fully actuated systems \cite{o2022neural} or studies the matched uncertainty setting \cite{boffi2021regret}. 
To the best of our knowledge, the work presented in \cite{lopez2021universal} is most relevant to this paper. The primary contribution of \cite{lopez2021universal} lies in the introduction of a novel direct adaptive control framework that effectively utilizes the certainty equivalence principle. This framework enables the design of stabilizing adaptive controllers for general nonlinear systems with unmatched uncertainties. 

The authors of the aforementioned work achieve the above goals by first defining the unmatched \ac{clf}, which is a family of \ac{clf}’s parameterized over all possible models. Subsequently, the approach proposed in ~\cite{lopez2021universal} dynamically adjusts the adaptation rate of the unknown parameter online to mitigate the influence of estimation transients on stability. However, it is important to note several restrictive assumptions associated with this method. 
Firstly, the proposed adaptive control framework specifically caters to nonlinear systems with constant unknown parameters. Secondly, it relies on the existence of a family of differentiable Lyapunov functions or a smooth manifold $\mathcal{M}$ that describes the contracting region for the system with any parameter estimates at all times $t$. These assumptions impose limitations as the existence of a Lyapunov function for a nonlinear system is generally non-trivial, and the requirement of differentiable Lyapunov functions for any variable estimate is highly restrictive. Additionally, recomputing the adaptive reference model for each new parameter estimate adds complexity and computational overhead. While the work in \cite{lopez2021universal} presents an innovative adaptive control framework, the aforementioned assumptions and limitations restrict its feasibility for certain applications, particularly in dynamic environments where time-varying uncertainties are prevalent.

\subsection{Statement of contribution}
\label{subsec:contributions}

For a complex nonlinear dynamical system with an initial robust controller, we aim to study the “effectiveness” of a data-driven feedback controller as it relates to \textbf{finite-time} regulation of system states in the feasible state space in the presence of unknown time-varying uncertainty. 
The contribution of the proposed work is thereby as follows: (1) we expand the notion of ``regularizability'' introduced in \cite{talebi2021regularizability} for a class of partially unknown nonlinear systems, where we model the time-varying uncertainty as the unknown part of the system dynamics.
(2) By using linear matrix inequalities (LMIs), we characterize rapid-regularizability for a linear time-varying (LTV) representation of the nonlinear system with locally-bounded higher-order (nonlinear) terms. We then clarify how this notion relates to the spectral properties of the underlying LTV system.
(3) In the second part of this paper, we introduce the \acf{DG-RAN} algorithm. This online iterative synthesis procedure utilizes the discrete time-series data of a single trajectory to regulate system states while generating informative data for identifying the dynamics of the disturbance. 

\subsection{Outline}
The rest of the paper is organized as follows. 
We introduce the problem statement in \Cref{sec: Prob-statement}.
The notion of rapid-regularizability for a partially unknown nonlinear system is then introduced in \S\ref{sec:RRN}. The properties of such systems are further studied in the subsequent part of this section. In \S\ref{sec:DG-RAN} the \ac{DG-RAN} algorithm is proposed as the means of online regulation of partially-unknown nonlinear systems.
An illustrative
example is presented in \Cref{sec: experiment} followed by concluding remarks in \Cref{sec: conclusion}.

\section{Problem Statement}
\label{sec: Prob-statement}

In this paper, we examine the control of an uncertain dynamical systems that has been
set to follow a predefined trajectory. In this direction, consider the following uncertain nonlinear dynamical system,
\begin{equation}
    \label{Eq:nonlinear_odd}
    \dot{\x}(t)=f(\x(t), \u(t)) + F\dot{\theta}(t), \quad t \in\left[t_{0}, t_{f}\right] ,
\end{equation}
where $\x(t) \in \mathbb{R}^{n_x}$ is the system states, $\u(t) \in \mathbb{R}^{n_u}$ is the control input, $F\in \mathbb{R}^{n_x\times n_\theta} $ is a known basis function for the environmental disturbance and $\theta(t) \in \mathbb{R}^{n_\theta}$ is an unknown time-varying variable. Using handpicked or learned basis functions, systems with non-parametric or nonlinearly parameterized uncertainties can be transformed into \Cref{Eq:nonlinear_odd}. One example of such a system is an airplane ascending and landing in the presence of unknown wind. We assume the $f(\x(t), \u(t))$ is locally Lipschitz uniformly and at least once differentiable, with measurable state $\x$.
Although the dynamics \Cref{Eq:nonlinear_odd} is continuous, we can only measure system states and control inputs at discrete times. Hence we consider the following assumption for measuring data.
\begin{assumption}
     The system states can be measured at constant (positive) intervals of length $\delta t \in \mathbb{R}$. Hence at time $t$, we have a streaming data set ${\cal D}_t = \{x_i, u_i\}_{i=0}^{n_t}$ of size $n_t = \lfloor t/\delta t \rfloor$.
\end{assumption}

Assume that $\{\bar{\x}(t), \bar{\u}(t)\}_{t=t_{0}}^{t_{f}}$ is the nominal (reference) trajectory that satisfies the dynamics in the absence of uncertainty, for some initial condition $\x_0$, and define,
\begin{subequations}
    \label{Eq:difference-variables-odd}
    \begin{align}
        \eta(t) &\coloneqq \x(t)-\bar{\x}(t), \label{sebEq:eta}\\
        \widetilde {\theta}(t) & \coloneqq  \theta(t) - \hat{\theta}(t), \label{Eq:sub-tilde-theta}\\
         \xi(t) + \z(t, {\cal D}_t) &\coloneqq \u(t)-\bar{\u}(t), 
    \end{align}
\end{subequations}
where $\eta(t)$ is the deviation from the nominal trajectory, $\hat{\theta}(t)$ is the estimate of the variable $\theta(t)$, and $\widetilde {\theta}(t)$ is the estimation error at time $t$. We can rewrite \Cref{Eq:nonlinear_odd} in terms of the variables in \Cref{Eq:difference-variables-odd} by using a first-order Taylor series expansion around the nominal trajectory as
\begin{equation}
\label{Eq:first-order-taylor}
    \begin{aligned}
        \dot{\eta}(t) = A(t) \eta(t) + B(t)& \xi(t)  + \delta g\left(\x(t), \u(t)\right) + B(t) \z(t, {\cal D}_t) + F\dot{\theta}(t),
    \end{aligned}
\end{equation}
% + ~\Pi_{\mathcal{R}(B)^\perp}  F \theta + ~\Pi_{\mathcal{R}(B)} F \theta + B(t) z(t) \\
%         &= A(t) \eta(t) + B(t) \xi(t) + g(x(t), u(t)) + ~\Pi_{\mathcal{R}(B)^\perp} F \theta + B ( B^\dagger F \theta + z(t))
where $A(t)$ and $B(t)$ are the partial derivatives of $f$ evaluated along the nominal trajectory, and $\delta g\left(\x(t), \u(t) \right) = g(\x(t),\u(t)) - g(\bar{\x}(t), \bar{\u}(t) )$ represents the higher order (nonlinear) terms. 

\begin{assumption}
    The higher order (nonlinear) term, $g \left(\x(t), \u(t)\right)$, is locally Lipschitz bounded. 
\end{assumption}

The general practice is to design a robust controller (with state denoted by $\xi(t)$) offline by considering the uncertainty that is known a priori \cite{rahimi2022robust}. However, in the online setting, the robust controller might fail when the effects of environmental uncertainties are not captured by the nominal uncertainty. Therefore, one needs an online data-driven approach (with state denoted by $\z(t, {\cal D}_t)$) to update the control policy to satisfy operational constraints.

Precisely, we assume that the offline robust controller is designed using the Lyapunov theory. Suppose that $\mathcal{X}_\mathcal{F} \in \mathbb{R}^{n_x}$ and $\mathcal{U}_\mathcal{F} \in  \mathbb{R}^{n_u}$ are the (possibly nonconvex) sets of feasible state and control vectors. A Lyapunov function can be used to seek out nearby feasible trajectories around the nominal $\{\bar{\x}(t), \bar{\u}(t)\}_{t=t_{0}}^{t_{f}}$. In other words, denote a funnel by 
$$\mathcal{F}(t) \subseteq \mathcal{X}_\mathcal{F} \times \mathcal{U}_\mathcal{F}.$$
The set $\mathcal{F}(t)$ is comprised of time-varying state and control trajectories 
that are invariant and contained inside the respective feasible regions.
%
% The invariance property of a funnel means that if a particular initial condition is inside the entry of the funnel (at some initial time $t_0$), then the entire subsequent trajectory remains inside the funnel as well.

For the original system dynamic $\dot{\x}(t) = f(\x(t), \u(t))$, one can use the Lyapunov theory to design a robust feedback controller $\xi(t)$, which guarantees the invariance of $\mathcal{F}(t)$; however, such a control policy might fail in the presence of environmental uncertainties, thus motivating our contributions in this work.

\subsection{Funnel Synthesis}
\label{sec:funnel-synthesis}

In this section, we briefly describe the funnel synthesis \cite{Reynolds2021FunnelProblem} approach, which is based on the notion of quadratic stability as defined in \cite{Boyd} and \cite{accikmecse2011robust, accikmes2008stability, d2013incremental}. 
For a fixed reference trajectory $\{\bar{\x}(t), \bar{\u}(t)\}_{t=t_{0}}^{t_{f}}$, the error dynamic for $\dot{\x}(t) = f(\x(t), \u(t))$ can be equivalently expressed using structured nonlinearities as follows
\begin{equation}
    \label{Eq:structured-nonlinearlity-odd}
    \dot{\eta}(t) =A(t) \eta(t)+B(t) \xi(t) + E \delta p(t)
\end{equation}
where $\eta(t)$ is defined in \Cref{sebEq:eta}, $\xi(t) = K(t) \eta(t)$, and the pair $(\delta q(t),\delta p(t)) \in \mathbb{R}^{n_{q}} \times \mathbb{R}^{n_{p}}$ capture the higher order terms through the nonlinear functions $\phi_i : \mathbb{R}^{n_{q_i}} \rightarrow \mathbb{R}^{n_{p_i}}$.
\begin{subequations}
\label{Eq:structure}
    \begin{align}
        \delta q(t) &= q(t) - \bar{q}(t) = \left(H + G K(t)\right) \eta(t) \, , \quad \text{where} \quad \bar{q}(t) = H \bar{x} + G\bar{u}  \\
        \delta p(t) &= \phi\left(t,q(t)\right) - \phi\left(t,\bar{q}(t)\right)  \, ,
    \end{align}
\end{subequations}
The details of this transformation are laid out in \cite{TaylorPatrickReynolds2020ComputationalSystems}.

The goal is to seek the largest possible funnel that satisfies input and state constraints. This allows us to implicitly define a large family of trajectories by using the functions that define the funnel, thereby providing the ability to guarantee the availability of a feasible trajectory over a larger region of parameter variations.
Particularly, consider the scalar-valued function $V: \mathbb{R}^{n_{x}} \rightarrow \mathbb{R}$ defined by
\begin{equation}
V(\eta(t))=\eta(t)^{\intercal} Q(t)^{-1} \eta(t)
\end{equation}
where $Q(t) \in \mathbb{S}_{++}^{n_{x}}$ is a matrix-valued function of time whose range space lies in the set of positive definite matrices. As a result, we have $V(\eta(t))>0$ for all $t \in\left[t_{0}, t_{f}\right]$ whenever $\eta(t) \neq 0$.

Having introduced each of the time-varying terms, we henceforth omit the argument of time ``$t$'' whenever possible. The 1-level set of $V(\eta(t))$ is the set of states that satisfy the quadratic inequality $\eta^{\intercal}  Q^{-1} \eta  \leq 1$, which is also the equation of a non-degenerate $n_{x}$-dimensional ellipsoid. We denote the ellipsoid defined by the positive definite matrix $Q$ and centered at the origin as
\begin{equation}
\label{Eq:ellipsoid-def}
    \mathcal{E}_{Q} \coloneqq \left\{\eta \in \mathbb{R}^{n_{x}} \mid \eta^{\intercal} Q^{-1} \eta \leq 1\right\}
    % =\left\{Q^{1 / 2} v \mid\|v\|_{2} \leq 1\right\}
\end{equation}
If $\eta \in \mathcal{E}_Q$, then $C\eta \in \mathcal{E}_{CQC^\intercal}$, a fact that can be proven easily via Schur complements when $C$ is full row-rank\footnote{When $C$ is not full row-rank, the ellipsoid $\mathcal{E}_{CQC^\intercal}$ is a degenerate ellipsoid.}.

We can now formally define a quadratic funnel as described in \cite{Reynolds2021FunnelProblem}.
\begin{definition}
    (Quadratic Funnel). A quadratic funnel, $\mathcal{F}$, is a set in state and control space that is parameterized by a time-varying positive definite matrix $Q(t) \in \mathbb{S}_{++}^{n_{x}}$ and a time-varying matrix $K(t) \in \mathbb{R}^{n_u\times n_x}$. Specifically, we have
    \begin{equation}
        \label{Eq:quadratic-funnel}
        \mathcal{F} = \mathcal{E}_Q \times \mathcal{E}_{KQK^\intercal}, \quad \mathcal{E}_Q \subseteq \mathcal{X}_\mathcal{F}, \quad \mathcal{E}_{KQK^\intercal} \subseteq \mathcal{U}_\mathcal{F}
    \end{equation}
    $K$ is called the correction law associated with the quadratic funnel.
\end{definition}

The authors in \cite{Reynolds2021FunnelProblem} showed that the following optimization renders the maximum feasible quadratic funnel in the state and control space, where a time-varying feedback controller of the form $\xi(t) = K(t) \eta(t)$ guarantees convergence to the nominal trajectory, starting anywhere inside the funnel.

\begin{problem}
\label{Prob:funnel_synthesis}
    (Quadratic Funnel Synthesis \cite{Reynolds2021FunnelProblem}). Given a nominal trajectory $\{\bar{x}(t), \bar{u}(t)\}_{t=t_0}^{t_f}$ that satisfies the nonlinear dynamics $\dot{x}(t) = f(x(t), u(t))$, an appropriate definition of the nonlinearity channels $(\delta q(t), \delta p(t)) \in \mathbb{R}^{n_{q}} \times \mathbb{R}^{n_{p}}$ and $\alpha > 0$, find the matrix-valued functions of time $Q(t), Y(t)$ and $M_\gamma(t)$ and scalar $\lambda(t)$ that solve the following optimization problem $\forall t \in\left[t_0, t_f\right]$.

    \begin{subequations}
        \label{Eq:subeq-ddc-funnel}
         \begin{align}
            \max _{ Q(\cdot), Y(\cdot), \lambda(\cdot), M_\gamma(\cdot)}\hspace{1.3cm} & \log \operatorname{det} Q\left(t_0\right) \tag{\theparentequation}\\
            \textbf{subject to } \hspace{1.8cm} & \nonumber\\
            & 0 \preceq Q \preceq Q_{\max },\; \;  0 \leq \lambda,\; \; M_\gamma \in \mathcal{M}_{\phi, \Omega}, \\
            & \begin{bmatrix}
                \mathbf{Q} {A}^\intercal +  {A}\mathbf{Q}  + T_1 & E+\lambda \mathbf{Q}{C_{cl}}^{\intercal} M_{12} \\
                E^{\intercal}+\lambda M_{12}^{\intercal} {C_{cl}}^{\intercal} \mathbf{Q} & \lambda M_{22}
            \end{bmatrix} \preceq 0, \\
            & \begin{bmatrix}
                Q & Y^{\intercal} \\
                Y & R_{\max }
            \end{bmatrix} \succeq 0,
         \end{align}
    \end{subequations}
    where $\Omega = \mathcal{E}_{C_{cl} Q C_{cl}^{\intercal}}$, $M_\gamma = \text{diag}\left([\gamma^2I, -I]\right)$, $C_{cl} = H+G\mathbf{K}$, $\mathbf{Y} = K \mathbf{Q}$, and $T_1  = -\dot{\mathbf{Q}}  + \alpha \mathbf{Q} + \mathbf{Y}^\intercal B^\intercal + B\mathbf{Y} + \lambda \mathbf{Q}C_{cl}^\intercal M_{11}C_{cl}\mathbf{Q}$. 
\end{problem}

\begin{remark}
    The ellipsoids $\mathcal{E}_{Q_{\max}}$, and $\mathcal{E}_{R_{\max}}$ defined by the positive definite matrices $Q_{\max}\succ 0 $ and $R_{\max}\succ 0 $ are respectively the maximum volume ellipsoids in the feasible state and control space.
\begin{align*}
    \mathcal{E}_{Q_{\max}} \subseteq \mathcal{X}_{\mathcal{F}}, \quad \mathcal{E}_{R_{\max}} \subseteq \mathcal{U}_{\mathcal{F}}
\end{align*}

Furthermore, as defined in \cite{Reynolds2021FunnelProblem}, $\mathcal{M}_{\phi, \Omega}$ is the set of local multiplier matrices for the nonlinear mapping $\phi: \mathbb{R}^{n_{q}} \rightarrow \mathbb{R}^{n_{p}}$ over the set $\Omega$, and any symmetric matrix $M \in \mathbb{S}^{(n_q+n_p)}$ such that $M \in \mathcal{M}_{\phi, \Omega}$, satisfies the following inequality.

\begin{equation}
    \begin{bmatrix}
        q \\
        \phi(q)
    \end{bmatrix}^{\intercal} M \begin{bmatrix}
        q \\
        \phi(q)
    \end{bmatrix} \geq 0, \quad \text { for all } q \in \Omega .
\end{equation}
\end{remark}

\section{Rapidly-Regularizable Nonlinear Systems}
\label{sec:RRN}

In this section, we first introduce the notion of ``rapid-regularizablity'' for a class of nonlinear dynamical systems. We then provide the necessary and sufficient conditions for the corresponding data-driven controller and conclude the section with the main theorem.

 \begin{definition}
 \label{def:rapidly-regularizable}
     Denote $\bar{Q}(t)$ and $K(t)$ for $t \in [t_0, t_f]$ as the solution to Problem \ref{Prob:funnel_synthesis}. The uncertain dynamical system of the form \Cref{Eq:first-order-taylor} is called ``rapidly-regularizable'' if there exists a linear feedback controller $z(t, \theta) = \tK_*(t) \theta(t)$ such that for all $t \in [t_0, t_f]$ the following holds.
     \begin{equation}
         \eta(t) \in \mathcal{E}_{\bar{Q}(t)}, \quad
         \xi(t) + z(t, \theta) \in \mathcal{U}_{\mathcal{F}}
     \end{equation}
     where $\xi(t) = K(t) \eta(t)$.
 \end{definition}

By definition, depending on the dynamics of the variable $\theta(t)$, there should exist a feedback policy $z(t, \theta)$ that further assists the existing robust controller $\xi(t)$ in keeping the system in the invariant state space $\mathcal{E}_{\bar{Q}} \subseteq \mathcal{X}_{\mathcal{F}}$. The following theorem provides a sufficient condition for the nonlinear system of the form \Cref{Eq:first-order-taylor} to be rapidly regularizable\footnote{We will provide additional discussion on the limitations of the control space $\mathcal{U}_{\mathcal{F}}$ in the final manuscript.}.

\begin{theorem}
\label{Thm:RRS}
    Denote $\bar{Q}(t),\; \bar{Y}(t) = \mathbf{K}(t) \bar{Q}(t)$ and $M_\gamma(t)$ for $t \in [t_0, t_f]$ as the solution to Problem \ref{Prob:funnel_synthesis}. Suppose the dynamic of $\theta(t)$ is known and $\dot{\theta}(t) = \tilde{S}\theta(t)$ where $\theta(t=0)=\theta_0$. If there exists a feedback controller of the form $z(t,\theta) = {\tK}_*(t) \theta(t)$, a matrix-valued function of time $\tilde{M}(t)$ and scalars $\lambda(t)$, $\nu(t), \beta(t)$ that for all time $t \in [t_0,t_f]$ solves \Cref{Eq:subeq-RapidR}, then the uncertain dynamical system of the form \Cref{Eq:first-order-taylor} is ``rapidly-regularizable''.
     \begin{subequations}
    \label{Eq:subeq-RapidR}
        \begin{align}
            \min_{{\tK}_*(\cdot), \tilde{M}(\cdot), \nu(\cdot), \beta(\cdot), \lambda(\cdot)}   \hspace{1.3cm} &\sigma_*(\cdot)   \tag{\theparentequation}\\
            \textbf{subject to }  \hspace{2cm} &\nonumber  \\
            & 0 < \beta,\; \;  0 < \nu,\; \;  0 < \lambda,\; \; \tilde{M} \in \mathcal{M}_{\phi, \Psi}, \\
            & \begin{bmatrix}
                \bar{Q}{A}^\intercal  + A \bar{Q} + T_2 &  \gamma E & \nu \tilde{\gamma} E  &[ \qquad T_3 &] \\
                 \gamma E^\intercal & -\beta I & 0 & \qquad 0 \quad 0 & \\
                 \nu \tilde{\gamma} E^\intercal & 0 & -\nu I  & \qquad 0 \quad 0 &\\
                * & * & *  &[ \qquad T_4 &]
            \end{bmatrix} \preceq 0  , \\
                & z(t,\theta) + \xi(t) \in \mathcal{U}_{\mathcal{F}} ,
        \end{align}
        \end{subequations}
        where $T_2 \coloneqq  \mathbf{Y}^\intercal B^\intercal + B \mathbf{Y} - \dot{\bar{Q}} + \lambda \bar{Q} + \beta \bar{Q}C_{cl}^\intercal C_{cl}\bar{Q}$, $\;\, T_3 \coloneqq \begin{bmatrix}
                 B {\tK}_* \theta + F \tilde{S} \theta & 0
            \end{bmatrix}$, and $\;\, T_4 \coloneqq \begin{bmatrix}
                -\lambda \sigma_* &\theta^\intercal {\tK}_*^\intercal G \\
                G {\tK}_* \theta & -\nu I
            \end{bmatrix}$.
\end{theorem}

\begin{proof} 
    For the uncertain system \cref{Eq:first-order-taylor}, consider the Lyapunov function $V_{\theta}(t) = \eta(t, \theta)^\intercal \bar{Q}^{-1}(t) \eta(t, \theta)$, defined by the matrix-valued function of time $\bar{Q}(t)$. By the Lyapunov theory, if  $z(t,\theta) = {\tK}_*(t) \theta(t)$ satisfies conditions in \Cref{Eq:lyaponov-stability}, then the system states remain inside the funnel. 
    
    % we claim that if  is rapidly-regularizable, then there exists  that further assists the robust controller in keeping the system states $\eta(t)$ inside the feasible state space $\mathcal{E}_{\bar{Q}(t)}$. We claim that if the $z(t,\theta)$ satisfies   (If rapidly-regularizable then k exists) Hence, for such system the should satisfy the following
    %
    \begin{equation}
    \label{Eq:lyaponov-stability}
        \textstyle \dot{V}_{\theta}(t)  < 0, \quad
        \text{for all }  \; \sigma_*(t) \leq V_{\theta}(t),\quad \begin{bmatrix}
            \delta\tilde{p}(t)\\
            \tilde{q}(t) 
        \end{bmatrix}^{\intercal} \tilde{M}(t)\begin{bmatrix}
            \delta\tilde{p}(t) \\
            \tilde{q}(t)
        \end{bmatrix} \geq 0
    \end{equation}
    We again used structured nonlinearities and defined the set of nonlinear functions $\tilde{\phi}_i: \mathbb{R}^{n_{\tilde{q}_i}} \rightarrow \mathbb{R}^{n_{\tilde{p}_i}} $ such that
    \begin{subequations}\label{Eq:updated-nonlinearities}
        \begin{align}
            \tilde{q}(t) &\coloneqq G_i \tilde{K}_*(t){\theta}(t) \\
            \delta\tilde{p}(t) &\coloneqq \tilde{\phi}(t,\tilde{q}(t)) = \phi(t,q(t) + \tilde{q}(t) )-\phi(t,q(t)) \label{Eq:tilde_p}
        \end{align}
    \end{subequations}
    where $q(t)$ and $\delta p(t)$ are defined in \Cref{Eq:structure}, and by the definition of $\delta\tilde{p}(t)$ we have 
    $$\delta g(\x(t),\u(t)) = \phi(t,q(t)+\tilde{q}(t)) - \phi(t,\bar{q}(t)) = \delta p(t) + \delta \tilde{p}(t) \, .$$
     Let $\chi = \left[ 
            \eta,\, p ,\, \tilde{p} ,\, \mathbf{1}
        \right]^\intercal$. Expanding \Cref{Eq:lyaponov-stability} in terms of the close loop system leads to

    \begin{equation}
        \label{Eq:Lyaponov-stability-proof}
            \chi^\intercal
            \begin{bmatrix}
                {A}^\intercal \bar{Q}^{-1} + \bar{Q}^{-1}  A + T_5 & \bar{Q}^{-1}E & \bar{Q}^{-1}E & \bar{Q}^{-1} (B {\tK}_* \theta + F\tilde{S} \theta )\\
                * & 0 & 0 & 0
            \end{bmatrix} \chi   \leq 0,
    \end{equation}
    for all:
    \begin{align*}
          \begin{bmatrix}
             \eta \\ \mathbf{1}
         \end{bmatrix}^\intercal \begin{bmatrix}
             \bar{Q}^{-1} & 0\\
             0& -\sigma_*
         \end{bmatrix} \begin{bmatrix}
             \eta \\ \mathbf{1}
         \end{bmatrix} > 0\;, \quad &
         \begin{bmatrix}
             \delta\tilde{p} \\ \tilde{q}
        \end{bmatrix}^\intercal  \tilde{M}  \begin{bmatrix}
             \delta\tilde{p} \\ \tilde{q}
        \end{bmatrix} \geq 0 \;, \quad  
        \begin{bmatrix}
             \delta q \\ \delta p
        \end{bmatrix}^\intercal  M_\gamma  \begin{bmatrix}
             \delta q \\ \delta p
        \end{bmatrix} \geq 0 \; ,
    \end{align*}
    where $T_5 \coloneqq \frac{d}{dt}\left(\bar{Q}^{-1}\right) + \mathbf{K}^\intercal B^\intercal \bar{Q}^{-1} + \bar{Q}^{-1} B \mathbf{K}$.
    We further consider $\tilde{M} = \text{diag}\left([-I, \; \tilde{\gamma}^{2}I]\right)$ as a valid local multiplier matrix, where $\tilde{\gamma}$ is a local Lipschitz constant for the nonlinear function $\tilde{\phi}$ over the set $\Psi$ .
    By using the S-procedure, the condition \Cref{Eq:Lyaponov-stability-proof} holds if and only if there exists scalars $\beta, \tilde{\beta}, \lambda \geq 0$ such that
    \begin{equation}
    \label{Eq:Lyapunov-stability-inv}
        \begin{bmatrix}
            {A}^\intercal \bar{Q}^{-1} + \bar{Q}^{-1}  A + T_5 & \gamma \bar{Q}^{-1} E & \tilde{\gamma}\bar{Q}^{-1} E & \bar{Q}^{-1} \left( B {\tK}_* \theta + F \theta\right) & 0 \\
             \gamma E^\intercal \bar{Q}^{-1} & -\beta I &  0 & 0 & 0\\ 
            * & *  &[& T_6 &]
        \end{bmatrix} \preceq 0 ,
    \end{equation}
    where we overuse the notation to define $T_5$, and define $T_6$ as
    \begin{align*}
        T_5 \coloneqq  \frac{d}{dt}\left(\bar{Q}^{-1}\right) + \mathbf{K}^\intercal B^\intercal \bar{Q}^{-1} + \bar{Q}^{-1} B \mathbf{K}+ \lambda \bar{Q}^{-1} + \beta C_{cl}^\intercal C_{cl}, \qquad T_6 \coloneqq \begin{bmatrix}
             -\lambda \sigma_* &\theta^\intercal {\tK}_*^\intercal G \\
            G {\tK}_* \theta & -\tilde{\beta}^{-1} I
        \end{bmatrix}.
    \end{align*}
    Pre- and post-multiplying \Cref{Eq:Lyapunov-stability-inv} by $\textbf{diag} \{\bar{Q}, I, \tilde{\beta}^{-1}, I , I \}$ and a change of variable $\nu = \tilde{\beta}^{-1}$ yields the final result.

\end{proof}

Assuming the disturbance dynamics is $\dot{\theta}(t) = \tilde{S}\theta(t)$ with the initial condition $\theta_0$, the system dynamic in \Cref{Eq:first-order-taylor} with control input $u(t) = \bar{u}(t)+\xi(t)+z(t, \theta)$, where $\xi(t) = \mathbf{K}(t) \eta(t)$, can be represented as,

\begin{equation}
\label{Eq:projection} 
    \dot{\eta}(t) = (A(t) + B(t) \mathbf{K}(t) )\eta(t) + \delta g(\x(t), \u(t)) + B(t)\left( B(t)^\dagger F\tilde{S}\theta(t) + z(t, \theta) \right) + \Pi_{\mathcal{R}(B(t))^\perp} F\tilde{S}\theta(t) .
\end{equation}
Setting $z(t, \theta) = -B^\dagger F \tilde{S} \theta(t) + \mathbf{\hat{K}}(t) \theta(t) = \tK_*(t) \theta(t)$, \Cref{Eq:projection} can be rewritten as 
\begin{equation}
\label{Eq:proj-dynamic}
    \dot{\eta}(t) = A_{cl}(t) \eta(t) + E \delta p(t) + E \delta\tilde{p}(t) + B(t)\mathbf{\hat{K}}(t) \theta(t) + \Pi_{\mathcal{R}(B(t))^\perp} F\tilde{S}\theta(t) .
\end{equation}
where $A_{cl}(t) = A(t) + B(t) \mathbf{K}(t)$, and $\delta p(t)$ and $\delta \tilde p(t)$ are defined in \Cref{Eq:structure} and \Cref{Eq:updated-nonlinearities}, respectively.  
Note that the signals $\Pi_{\mathcal{R}(B(t))^\perp} F\tilde{S}\theta(t)$ and $B(t)\mathbf{\hat{K}}(t) \theta(t)$ are now orthogonal. 
This implies that the control signal $z(t,\theta)$ would not directly affect the part of dynamics that is generated by $~\Pi_{\mathcal{R}(B(t))^\perp} F\tilde{S}\theta(t)$. Moreover, $z(t,\theta)$ causes additional disturbance in terms of $\delta g(\x(t), \u(t))$ captured by $\delta \tilde{p}(t)$.
The local Lipschitz bounds for $\delta g(\x(t), \u(t))$ drove from \Cref{Eq:structure} and \Cref{Eq:updated-nonlinearities} are only accurate when $\eta(t) \in \mathcal{E}_{\bar{Q}(t)}$. As such, in order to have even the possibility of achieving online performance for keeping the system states in the feasible region, we require that the part of the system described by $E \delta\tilde{p}(t) + B(t)\mathbf{\hat{K}}(t) \theta(t) + \Pi_{\mathcal{R}(B(t))^\perp} F\tilde{S}\theta(t)$ be bounded for all $t\in[t_0,t_f]$. 
This observation thereby motivates 
Proposition \ref{prop:RRS-property}. 
% Before we proceed, we make a remark on the limitations of the control space.
% %

% \begin{remark}
    
% \end{remark}

\begin{proposition}
    \label{prop:RRS-property}    
    The nonlinear system of the form \Cref{Eq:proj-dynamic} is rapidly regularizable in terms of \Cref{Thm:RRS}, if and only if 
        % \item The part of the system described by $\tilde{F}$, is stable (i.e. $\rho(\tilde{F}) \leq 0$), where
        % %
        % \begin{equation}
        % \label{Eq:tilde_F}
        %     \tilde{F} \coloneqq \Pi_{\mathcal{R}(\bar{B})^\perp} \left(\Pi_{\mathcal{R}(F^\intercal)} \tilde{S} \right)\, , \quad \text{where} \quad \mathcal{R}(\bar{B}) \coloneqq \bigcap\limits_{\tau\in[t_0,t_f]} \mathcal{R}(B(\tau))\, ,
        % \end{equation}
        % %
    the part of the system described by $\vartheta(t)$ is bounded such that $ 
        \|\vartheta(t)\|^2 < \sigma^*(t)$, 
    where
    \begin{subequations}
        \begin{equation}
              \quad \vartheta(t) = E \delta\tilde{p}(t) + B(t)\mathbf{\hat{K}}(t) \theta(t) + \Pi_{\mathcal{R}(B(t))^\perp} F\tilde{S}\theta(t)  \, , \quad \text{and} \quad \mathbf{\hat{K}}(t) = \tK_*(t) + B^\dagger F \tilde{S}
        \end{equation}
        Furthermore, the ellipsoid
        \begin{equation}
        \label{Eq:invariant-ellips}
            \mathcal{E}:=\left\{\eta: \eta^{\mathrm{T}} \bar{Q}^{-1} \eta \leq\|\vartheta(\cdot)\|_{\Delta t}^2\right\} \quad \text { with } \quad\|\vartheta(\cdot)\|_{\Delta t}=\sup_{t_f \geq t \geq t_0}\|\vartheta(t)\|
        \end{equation}
        is invariant and attractive.
    \end{subequations}
\end{proposition}

\begin{proof}
    
\end{proof}

\section{Data-Guided Regulation for Adaptive Nonlinear Control}
\label{sec:DG-RAN}

Suppose the dynamics of the disturbance $\theta(t)$ satisfies 
\begin{subequations}
\label{Eq:dynamic-theta}
    \begin{equation}
    \label{subEq:thet_contin} \tag{\theparentequation}
        \dot{\theta}(t) = \tilde{S}\theta(t), \quad \theta_0 \in \mathbb{R}^{n_\theta}
    \end{equation} 
    \begin{equation}
    \label{subEq:thet_disc} 
        \text{where } \quad \theta_{n_t} = \theta(n_t \delta t) = S^{n_t} \theta_0 \, , \quad \text{and} \quad S = \exp(\delta t \tilde{S})\, , \quad n_t\in\{0,1,2,\cdots\}
    \end{equation}
\end{subequations}
The primary focus of this section is devising an online, data-driven feedback controller to further assist the offline-computed robust controller for a rapidly-regularizable system with unknown disturbance dynamics \Cref{Eq:dynamic-theta}, i.e., $\tilde{S}$ and $\theta_0$ are unknown. To this end, we propose an iterative procedure for updating the feedback gain $\tK(t)$ at each time $t=n_t \delta t$. For a rapidly-regularizable system, We provide an optimization problem for obtaining the corresponding sub-optimal data-driven controller and conclude the section with the main theorem.

 Given a data set ${\cal D}_t = \{x_i, u_i\}_{i=0}^{n_t}$ available at time $t$, we are interested in characterizing the algorithmic map ${\cal D}_t \rightarrow \tK(t),$ where $\tK(t)$ is the data-driven regulating feedback gains for the uncertain dynamical system \Cref{Eq:first-order-taylor}. 
 In fact, we propose an iterative algorithm for updating the feedback gain (policy) based on the available (closed loop) data. The corresponding synthesis procedure is detailed in \Cref{alg:Controller}.

The goal is to use the measured data ${\cal D}_t = \{x_i, u_i\}_{i=0}^{n_t}$ to update estimates for $\theta_0$ and $S = \exp{(\tilde{S} \delta t)}$. Specifically, at time $``t_0"$, the algorithm measures $x_0$.
%
% and sets $\hat{\x}(0) = x_0$, where $\hat{\x}(t)$ represents the current estimate of the system states during each time interval $t\in [(n_t-1)\delta t, n_t \delta t]$.
% % 
% \begin{align*}
%     \label{Eq:hat-x}
%     \dot{\hat{\x}}(t) &\coloneqq f(\hat{\x}(t), \hat{\u}(t)) + F\dot{\hat{\theta}}(t), \quad \hat{\u}(t) \coloneqq \bar{u}(t) + \mathbf{K}(t)(\hat{\x}(t) - \bar{\x}(t)) + \tK(t) \hat{\theta}(t) \\
%     \dot{\hat{\theta}}(t) &\coloneqq \hat{\tilde{S}}_{n_t-1}\hat{\theta}(t), \quad \hat{S}_{n_t-1} \coloneqq \exp{(\hat{\tilde{S}}_{n_t-1} \delta t)}, \quad \hat{\theta}\left((n_t-1)\delta t\right) \coloneqq  \hat{\theta}_{n_t-1}
% \end{align*}
%
The ${n_t}^{th}$ measurement, $x_{n_t}$, is obtained at $t=n_t\delta t$, where $n_t \in \{1,2,\cdots\}$. We define
\begin{equation}
\label{Eq:observation}
    y_{n_t} \coloneqq F^\dagger \left(\Delta x_{n_t} - \Delta f_{n_t} \right), \; 
\end{equation}
where $\; \Delta x_{n_t} \coloneqq x_{n_t} - x_{n_t-1}\;$ and $\; \Delta f_{n_t} \coloneqq \frac{\delta t}{2}\left(f(x_{n_t}, u_{n_t}) + f(x_{n_t-1}, u_{n_t-1})\right) .$
When $n_t > 1$, the algorithm updates the estimates for $\theta_0$ and $S = \exp{(\tilde{S} \delta t)}$ as follows:
\begin{subequations}
    \label{Eq:estimate-theta-S}
    \begin{equation}
    \textstyle
        \hat{S}_{n_t} = \mathcal{Y}_{n_t}\mathcal{Z}_{n_t-2}^\dagger   , \; \qquad
        \hat{\theta}_{n_t}^{\circ} =  \left(\hat{S}_{n_t} - I\right)^\dagger y_1, \tag{\theparentequation}
    \end{equation}
    where 
    \begin{align}
    \textstyle 
        \mathcal{Y}_{n_t} &= \begin{bmatrix}
            y_2 & y_3 - \mathcal{Y}_{2} \mathcal{Z}_0^\dagger y_2 & \cdots &  y_{n_t} - \mathcal{Y}_{n_t-1} \mathcal{Z}_{n_t-3}^\dagger y_{n_t-1}
        \end{bmatrix} \;  , \\
        \mathcal{Z}_{\ell} &= \begin{bmatrix}
            \z_0 & \z_1 & \cdots &  \z_{\ell}
        \end{bmatrix} \; , \quad
          \z_\ell = \Pi_{\mathcal{R}\left(\mathcal{Z}_{\ell-1}\right)^\perp} y_{\ell+1} \; , \quad \text{for } \ell\geq 1 \;, \quad \z_0 = y_1 \; , \label{subEq:zk-wk}
    \end{align}
\end{subequations}
and the columns of $V_r$ are the right singular
vectors of the ``thin" SVD of $F = U_{r} \Sigma_{r} V_{r}^\intercal$, where $r$ denotes the dimension of the range of matrix $F$. 

\begin{algorithm}[tp]
		\caption{\acf{DG-RAN}}
		\begin{algorithmic}[1]
			\State \textbf{Initialization} \hspace{1mm} (at $t=0$)
			\State \hspace{2mm} Fix $\delta t$; Measure $\x_0$; set $\Tilde{K} = 0$ and $n_t=0$
			\State \hspace{2mm} Set \hspace{1mm}$\mathcal{Y}_{n_t} = \left(\begin{array}{c} \ \end{array}\right)$ ,  and  \; $\mathcal{Z}_{n_t} = \left(\begin{array}{c} \ \end{array}\right)$ 
			\State \textbf{while stopping criterion not met}\footnotemark
			\State \hspace{2mm} \textbf{If } $n_t < \lfloor t / \delta t \rfloor$
                \State \hspace{5mm} Set \hspace{1mm} $n_t = \lfloor t / \delta t \rfloor$ 
			\State \hspace{5mm} Measure \hspace{1mm} $\x_{n_t}$
                \State \hspace{5mm} Set \hspace{1mm} $y_{n_t} = F^\dagger \left(\Delta x_{n_t} - \Delta f_{n_t} \right)$  
                \State \hspace{5mm} \textbf{If } $n_t == 1$
                \State \hspace{7mm} Set \hspace{1mm} $\mathcal{Z}_{0} =
	    \left(\begin{array}{c}
		y_1 \end{array}\right)$
                \State \hspace{5mm} \textbf{else: } 
                \State \hspace{7mm} Update \hspace{1mm} $\mathcal{Y}_{n_t} =
	    \left(\begin{array}{cc}
		\mathcal{Y}_{n_t-1} & y_{n_t} - \mathcal{Y}_{n_t-1} \mathcal{Z}_{n_t-3}^\dagger y_{n_t-1}
		\end{array}\right)$                
		\State \hspace{19.5mm} $\mathcal{Z}_{n_t-1} =
	    \left(\begin{array}{cc}
		\mathcal{Z}_{n_t-2} & \Pi_{\mathcal{R}\left(\mathcal{Z}_{n_t-1}\right)^\perp} y_{n_t}
		\end{array}\right)$
    	\State \hspace{20mm} $\hat{S}_{n_t} = \mathcal{Y}_{n_t}\mathcal{Z}_{n_t-2}^\dagger , \quad \hat{\theta}_{n_t}^{\circ} =  \left(\hat{S}_{n_t} - I\right)^\dagger y_1$
			\State \hspace{7mm} Update \hspace{1mm} $\tK(t)$ by solving \Cref{Eq:subeq-regularizability}
		\end{algorithmic}
		\label{alg:Controller}
	\end{algorithm}

\footnotetext{The stopping criterion can be application specific.
For instance, for identifying modes of the disturbance system affecting the dynamics, generating $k_1$ linearly independent data is sufficient.}

In what follows, we provide the performance analysis for \ac{DG-RAN} in the general setting. 
\ac{DG-RAN} is particularly relevant when $n_t\leq r$, where $r$ denotes the dimension of the underlying disturbance system $\dot{\theta}(t)$ affecting system states $x(t)$. 
We examine the effects of \ac{DG-RAN} in terms of the informativity of generated data in \Cref{subsec:informativity}, and provide convergence analysis of the algorithm in \Cref{subsec:convergence}.

\subsection{Informativity of the DG-RAN Generated Data }
\label{subsec:informativity}

In the sequel, we show that DG-RAN generates linearly independent data, describing the observable modes of the disturbance dynamics $\left(\Pi_{\mathcal{R}(F^\intercal)} S \right)$ from the state-trajectory measurements ${\cal D}_t = \{x_i, u_i\}_{i=0}^{n_t}$. We refer to this as the informativity of data, and
then proceed to make a connection between this independence
structure and the number of excited modes of the disturbance dynamics $\left(\Pi_{\mathcal{R}(F^\intercal)} S \right)$ affecting the system states

\begin{lemma}
    \label{lem:algo-dyn}
    Let $\theta_0$ excites $k_1+k_2$ modes of $S$ such that $k_1$ modes are in $\mathcal{R}(F^\intercal)$ and $k_2$ modes are in $\mathcal{R}(F^\intercal)^\perp$. For all $n_t > 2$,  $\{\z_0, \z_1, \cdots,\z_{n_t-3}\}$ is a set of ``orthogonal'' vectors {(possibly including the zero vector)}.
    Furthermore, If the excited modes correspond to distinct eigenvalues of $S$ where $\lambda_i\not=1$ for $\v_i\in\mathcal{R}(F^\intercal)$, then  $\left[ y_1, \; \mathcal{Y}_{n_t}\right]$, generated by \Cref{alg:Controller}, is a set of linearly independent vectors for any $0 < m$ such that $ m \leq k_1 \leq r$.
\end{lemma}

\begin{proof}
    The definition of $\z_\ell$ in \Cref{subEq:zk-wk} implies that $\z_\ell \perp \mathcal{R}(\mathcal{Z}_{\ell-1})$ for all $\ell \geq 1$, and $\z_i \in \mathcal{R}(\mathcal{Z}_{\ell-1})$ for all $i=1,\ldots,\ell-1$ and all $\ell \geq 1$.
    Thus, the first claim follows since $\{\z_0,\z_1,\cdots,\z_\ell\}$ consists of orthogonal vectors. For the second claim, without loss of generality, let $\lambda_1,\dots, \lambda_{k_1}$ be the eigenvalues of $S$ corresponding to the excited modes $\v_1, \dots, \v_{k_1} \in \mathcal{R}(F^\intercal)$, and similarly $\lambda_{k_1+1},\dots, \lambda_{k_1+k_2}$ be corresponding to $\v_{k_1+1}, \dots, \v_{k_1+k_2} \in \mathcal{R}(F^\intercal)^\perp$.  For $n_t\geq 1$ the trajectory generated by \Cref{alg:Controller} satisfies
    \begin{equation}
        \label{Eq:traj-alg}
        \textstyle 
        \Delta x_{n_t} = \int_{(n_t-1)\delta t}^{n_t \delta t} f(\x(\tau), \u(\tau)) d\tau + F S^{n_t-1}\left(S-I\right) \theta_0 .
    \end{equation}
    By estimating the integral in \Cref{Eq:traj-alg}, we have
    \begin{align*}
    \textstyle 
        \left(F^\dagger F\right) (S)^{n_t-1}\left(S-I\right) \theta_0 = F^\dagger  \left(\Delta x_{n_t} - \Delta f_{n_t}  \right)+ \epsilon_{n_t} = y_{n_t}+ \epsilon_{n_t} ,
    \end{align*}
    where $\epsilon_{n_t} = F^\dagger \left(\int_{(k-1)\delta t}^{k \delta t} f(\x(\tau), \u(\tau)) d\tau - \Delta f_k\right)$ is assumed to be negligible\footnote{ A full discussion on the approximation error $\epsilon_{n_t}$ will be provided in the final manuscript.}. Note that $F^\dagger F = V_{r} V_{r}^\intercal = \Pi_{\mathcal{R}(V_r)}$, and  $\theta_0 = \sum\limits_{i=1}^{k_1+k_2} \rho_i \v_i$, where $\rho_i$ are some nonzero real coefficients.  Therefore, for the $n_t^{th}$ observation, we have
    \begin{align*}
        y_{n_t}  &=  \Pi_{\mathcal{R}\left(V_r\right)} (S)^{n_t-1}\left(S-I\right) \theta_0  = \Pi_{\mathcal{R}\left(V_r\right)} \left( \sum\limits_{i=1}^{k_1} \rho_i (\lambda_i)^{n_t-1}(\lambda_i-1) \v_i + \sum\limits_{i=k_1+1}^{k_2} \rho_i (\lambda_i)^{n_t-1}(\lambda_i-1) \v_i  \right) \; .
    \end{align*}
    Since $\mathcal{R}(V_r) = \mathcal{R}(F^\intercal)$ and $\v_{k_1+1}, \dots, \v_{k_1+k_2} \in \mathcal{R}(F^\intercal)^\perp$, we have
    \begin{align}
    \label{Eq: y_nt}
    \textstyle 
       y_{n_t} = \Pi_{\mathcal{R}\left(V_r\right)} (S)^{n_t-1}\left(S-I\right) \theta_0 =  \sum\limits_{i=1}^{k_1} \rho_i (\lambda_i)^{n_t-1}(\lambda_i-1) \v_i
    \end{align}

     Let $\w_1\coloneqq y_1 $ and $\w_2 \coloneqq y_2$, and $\w_{i}$ be the $i^{th}$ column of $\mathcal{Y}_{n_t}$. 
    We claim that for $n_t \geq 3$ there exist scalar coefficients $ \xi^{n_t}_{1}, \cdots, \xi^{n_t}_{n_t-1} \in \mathbb{R}$ such that, 
    \begin{equation}
    \label{Eq:w_nt-proof}
        \textstyle 
         \w_{n_t} \coloneqq y_{n_t} - \mathcal{Y}_{n_t-1} \mathcal{Z}_{n_t-3}^\dagger y_{n_t-1} = S \z_{n_t-2} =  \sum\limits_{i=1}^{k_1} \rho_i \left[ \left(\lambda_i\right)^{n_t-1} - \sum\limits_{j=1}^{n_t-2} \xi^{n_t-1}_j \left(\lambda_i\right)^{j}   \right](\lambda_i-1) \v_i,
    \end{equation}
    The proof of the last claim is by induction.
    By the definition of $\z_\ell$ in \Cref{subEq:zk-wk}, for $\ell\geq 1$ there exists scalar coefficients $ \zeta^\ell_{1}, \cdots, \zeta^\ell_{\ell} \in \mathbb{R}$ such that $\z_\ell = y_{\ell+1} - \sum\limits_{j=1}^{\ell} \zeta_j^\ell z_{j-1}$. For $n_t=3$, we have that,
    \begin{align*}
        \w_3 &= S y_2 - S y_1 y_1^\dagger y_{2} =  S \left( y_2 - \Pi_{\mathcal{R}\left(\mathcal{Z}_0\right)} y_2 \right) = S \z_1 = \textstyle \sum\limits_{i=1}^{k_1} \rho_i \left[ \left(\lambda_i\right)^2 -  \zeta_1^1 \left(\lambda_i\right)  \right](\lambda_i-1) \v_i,
    \end{align*}
    where the last equality is due to \Cref{Eq: y_nt}.
    By choosing $\xi_1^2 = \zeta_1^1 $, we have shown that \cref{Eq:w_nt-proof} holds for $n_t=3$.
    Now suppose that \cref{Eq:w_nt-proof} holds for all $3,\dots, n_t-1$; it now suffices to show that this relation also holds for $n_t$. By substituting the hypothesis for $3,\dots, n_t-1$ into \cref{Eq:w_nt-proof},
     \begin{align*}
         \w_{n_t} 
         = & \textstyle S y_{n_t-1} - S \mathcal{Z}_{n_t-3} \mathcal{Z}_{n_t-3}^\dagger y_{n_t-1} =  \textstyle S \left( y_{n_t-1} -\Pi_{\mathcal{R}\left(\mathcal{Z}_{n_t-3}\right)} y_{n_t-1}\right) = S \z_{n_t-2} \\
         & = \textstyle \sum\limits_{i=1}^{k_1} \rho_i \left(\lambda_i\right)^{n_t-1} (\lambda_i-1)\v_i  - \sum\limits_{j=1}^{n_t-2} \zeta_j^{n_t-2} S\z_{j-1}
         \\
         & = \textstyle \sum\limits_{i=1}^{k_1} \rho_i \left(\lambda_i\right)^{n_t-1} (\lambda_i-1)\v_i  - \sum\limits_{j=1}^{n_t-2} \zeta_j^{n_t-2} 
         \sum\limits_{i=1}^{k_1} \rho_i \left[ \left(\lambda_i\right)^{j} - \sum\limits_{\ell=1}^{j-1} \xi_{\ell}^{j-1} \left(\lambda_i\right)^{\ell}  \right] (\lambda_i-1)\v_i
     \end{align*}
     Therefore, $\w_{n_t} 
         = \textstyle \sum\limits_{i=1}^{k_1} \rho_i \left[ \star \right] (\lambda_i-1) \v_i $
    where $\star$ replaces the expression,
    \begin{equation*}
         \textstyle  (\lambda_i)^{n_t-1} - \sum\limits_{j=1}^{n_t-2} \zeta_j^{n_t-2} (\lambda_i)^{j} + \sum\limits_{j=1}^{n_t-2} \zeta_j^{n_t-2} \sum\limits_{\ell=1}^{j-1} \xi_{\ell}^{j-1}  (\lambda_i)^{\ell} .
     \end{equation*}
    By appropriate choices of $ \xi^{n_t-1}_{1}, \cdots, \xi^{n_t-1}_{n_t-2} \in \mathbb{R}$, we can rewrite $\star = \left(\lambda_i\right)^{n_t-1} - \sum\limits_{j=1}^{n_t-2} \xi^{n_t-1}_j \left(\lambda_i\right)^{j}  $.
    This completes the proof of the claim in \cref{Eq:w_nt-proof} by induction. 
    Now, let $\widehat{\w} = \sum\limits_{\ell= 1}^{m} \gamma_{\ell-1} \w_\ell$ for some $\gamma_{\ell-1} \in \mathbb{C}$ and some $m\leq k_1$.
    Then, by substituting $w_\ell$ from (\ref{Eq:w_nt-proof}) and exchanging the sums over $\ell$ and $i$ we have, 
    \begin{gather*}
    	\textstyle \widehat{\w} =
    	\sum\limits_{i=1}^{k_1}  \rho_i \left[ \gamma_0 + \gamma_1(\lambda_i)  + \sum\limits_{\ell=3}^{m} \gamma_{\ell-1} \left[  \left(\lambda_i\right)^{\ell-1} - \sum\limits_{j=1}^{\ell-2} \xi^{\ell-1}_j \left(\lambda_i\right)^{j}   \right] \right]  (\lambda_i-1)\v_i .
    \end{gather*}
    Now, by exchanging the sums over $\ell$ and $j$, it follows that,
    \begin{gather*}
    \textstyle	\widehat{\w} = \sum\limits_{i=1}^{k_1} \rho_i  \Big[ \gamma_0 + \sum\limits_{j=1}^{m-2} \big[ \gamma_{j} - \sum\limits_{\ell=j+1}^{m-1} \gamma_\ell \xi_{j}^{\ell} \big] (\lambda_i)^j + \gamma_{m-1}(\lambda_i)^{m-1} \Big](\lambda_i-1)\v_i  .
    \end{gather*}
    Since $\{\v_i\}_1^{k_1}$ are eigenvectors associated with distinct eigenvalues ($\lambda_i\not=1$), they are linearly independent. 
    Thus, noting that $\rho_i \neq 0$ for all $i = 1, \cdots, k_1$, then $\widehat{w}= 0$ implies that,
    \begin{equation*}
    \textstyle
         \gamma_0 + \sum\limits_{j=1}^{m-2} \big[ \gamma_{j} - \sum\limits_{\ell=j+1}^{m-1} \gamma_\ell \xi_{j}^{\ell} \big] (\lambda_i)^j + \gamma_{m-1}(\lambda_i)^{m-1}  = 0 ,
    \end{equation*}
    for all $i = 1, \dots,k_1$. By rewriting the last two sets of equations in matrix form, we get
    \begin{subequations}
        \begin{equation}\label{Eq:gamma-zero}
               L_{k_1}^m(A)(I-\Xi)  \ogamma= 0,  \tag{\theparentequation}
        \end{equation}
        where 
        \small
        \begin{equation} \label{eq:xi-gamma-def}
        \textstyle 
            L_{k_1}^m(S) \coloneqq {\left(
                \begin{array}{*4{c}}
                1& \lambda_{1} & \cdots & (\lambda_{1})^{m-1} \\
                1& \lambda_{2} & \cdots & (\lambda_{2})^{m-1} \\
                \vdots& \vdots & \ddots & \vdots  \\
                1& \lambda_{k} & \cdots & (\lambda_{k})^{m-1} \\
                \end{array}\right)} \; , \quad
            \Xi \coloneqq \left(\begin{array}{*6c}
                0 & 0 & 0 & 0 & \dots & 0\\
                0 & 0 &  \xi_1^2 & \xi_1^3 & \cdots & \xi_1^{m-1}\\
                0 & 0 & 0 &   \xi_2^3 & \cdots & \xi_2^{m-1}       \\
                \vdots & \vdots & \vdots & \vdots & \ddots &  \vdots      \\
                0 & 0 & 0 &0  & \cdots & \xi_{m-2}^{m-1}  \\
                0&0&0&0&\cdots&0
        		\end{array}\right), \quad \ogamma \coloneqq \left(\begin{array}{*4c}
        		\gamma_0\\
                \gamma_1\\
                \\
                \vdots\\
                \\
                \gamma_{m-1}
        		\end{array}\right).
        \end{equation}
        \normalsize
    \end{subequations}
    Note that $I - \Xi$ is invertible by construction,
    and $L_{k_1}^m(S)$ has a specific structure that hints at its invertibility.
    In fact, $L_{\ell}^{\ell}(S)$ is the Vandermonde matrix formed by $\ell$ eigenvalues of $S$ which would be invertible if and only if $\lambda_1, \cdots, \lambda_\ell$ are distinct.
    More generally, $L_\ell^r(S)$, where $r\leq \ell$, has full column rank if $\{ \lambda_1, \cdots, \lambda_\ell \}$ consists of $\ell$ distinct eigenvalues. Hence for $m \leq k_1$, $L_{k_1}^m(S)$ has full column rank and \cref{Eq:gamma-zero} implies that $\ogamma =0$ and thus $\{\w_1, \dots, \w_m\}$ is a set of linearly independent vectors.
    This observation completes the proof as $m\leq k_1$ was chosen arbitrarily.
\end{proof}

\begin{remark}
    Note that if $\left(F, S\right)$ is observable and $\theta_0$ excites all modes of $S$, then it is trivial to show that all modes of $S$ can be recovered  from the informative data. 
\end{remark}

\subsection{Convergence Property of the DG-RAN Algorithm }
\label{subsec:convergence}

For a rapidly regularizable system with unknown uncertainty dynamic of the form \Cref{Eq:dynamic-theta}, \Cref{thm:tilde_K} provides an optimization problem for obtaining the sub-optimal data-driven feedback $z(t, \mathcal{D}_t)$ at each time $t=n_t \delta t$.

\begin{theorem}
\label{thm:tilde_K}
    Denote $\bar{Q}(t)$ and $\bar{Y}(t) = \mathbf{K}(t) \bar{Q}(t)$ as the solution variables to Problem \ref{Prob:funnel_synthesis}, for $t \in [t_0, t_f]$. Further, assume the uncertain dynamical system of the form \Cref{Eq:first-order-taylor}  is rapidly-regularizable as defined in \Cref{def:rapidly-regularizable}. Then at time $t=n_t \delta t$, for the variable estimate $\hat{\theta}(t) = (\hat{S}_{n_t})^{t/\delta t} \hat{\theta}_{n_t}^{\circ}$ derived from \Cref{Eq:estimate-theta-S}, there exists a data-driven feedback controller of the form $z(\tau, \mathcal{D}_t) = \tK(\tau) \hat\theta(t)$, a matrix-valued function of time $\widehat{M}(\tau)$ and scalars $\lambda(\tau)$, $\nu(\tau), \beta(\tau)$ that for all time $\tau \in [t,t_f]$ solves the following optimization problem. 

    \begin{subequations}
    \label{Eq:subeq-regularizability}
        \begin{align}
            \min_{{\tK}(\cdot), \tilde{M}(\cdot), \nu(\cdot), \beta(\cdot), \lambda(\cdot)}   \hspace{1.3cm} &\sigma(\cdot)   \tag{\theparentequation}\\
            \textbf{subject to }  \hspace{2cm} &\nonumber  \\
            & 0 < \beta,\; \;  0 < \nu,\; \;  0 < \lambda,\; \; \tilde{M} \in \mathcal{M}_{\phi, \Psi}, \\
            & \begin{bmatrix}
                \bar{Q}{A}^\intercal  + A \bar{Q} + T_2 &  \gamma E & \nu \tilde{\gamma} E  &[ \qquad T_3 &] \\
                 \gamma E^\intercal & -\beta I & 0 & \qquad 0 \quad 0 & \\
                 \nu \tilde{\gamma} E^\intercal & 0 & -\nu I  & \qquad 0 \quad 0 &\\
                * & * & *  &[ \qquad T_4 &]
            \end{bmatrix} \preceq 0  , \\
                & z(\tau,\mathcal{D}_t) + \xi(\tau) \in \mathcal{U}_{\mathcal{F}} ,
        \end{align}
    \end{subequations}
    where $$\small T_2 \coloneqq  \mathbf{Y}^\intercal B^\intercal + B \mathbf{Y} - \dot{\bar{Q}} + \lambda \bar{Q} + \beta \bar{Q}C_{cl}^\intercal C_{cl}\bar{Q} ,\;\, T_3 \coloneqq \left[
             \nu \tilde{\gamma} E , \, B {\tK} \hat{\theta} + (\delta t)^{-1} F \ln{(\hat{S}_{n_t})} \hat{\theta} , \, 0\right], \; \text{and} \;\, T_4 \coloneqq \begin{bmatrix}
            -\lambda \sigma &\hat{\theta}^\intercal {\tK}^\intercal G \\
            G {\tK} \hat{\theta} & -\nu I
        \end{bmatrix}.$$
        \normalsize
\end{theorem}

\begin{proof} 
    Considering Lemma \ref{lem:algo-dyn}, at time $t=n_t \delta t$ we can rewrite the estimation for $\hat{S}_{n_t}$ as
    \begin{align*}
        \hat{S}{n_t} &= S \mathcal{Z}_{n_t-2}\mathcal{Z}_{n_t-2}^\dagger = \sum_{i=1}^{k_1} \left[ \Sigma (\Sigma-I) L_{k_1}^{n_t-1}(I-\Xi_{n_t})\right]_{i} \v_i ,
    \end{align*}
    where $L_{k_1}^{n_t-1}$ is defined in \Cref{eq:xi-gamma-def} and we define $\Xi_{n_t}$ by removing the first column and row of $\Xi$ in \Cref{eq:xi-gamma-def}. Note that generally, if $\{\lambda_1, \, \cdots , \,\lambda_{k_1}\}$, the eigenvalues of $S$ corresponding to the excited modes $\v_1, \dots, \v_{k_1} \in \mathcal{R}(F^\intercal)$,  consists of $r_d$ distinct eigenvalues (where $r_d \leq k_1$), then $L_{k_1}^{r_d}$ has full column rank. Hence, we can recover all modes of $\Pi_{\mathcal{R}(F^\intercal)} S $ when  $n_t = r_d + 1 $, and for $n_t < r_d+1$ we have
    \begin{equation}
    \label{Eq:hat_thet}
        \hat{\theta}(t) = \Pi_{\mathcal{R}\left(L_{k_1}^{n_t-1} \right) } \theta(t)
    \end{equation}
    Since the system is rapidly-regularizable in terms of \Cref{Thm:RRS}, we can conclude that for any estimate of $\hat{\theta}(t)$ \Cref{Eq:hat_thet}, there exists a controller of the form $z(\tau, \mathcal{D}_t) = \tK(\tau) \hat\theta(t)$ and a Lyapunov function $V_{\hat{\theta}}(\tau) = \eta(\tau, \hat{\theta}(t))^\intercal   \bar{Q}^{-1}(\tau)  \eta(\tau, \hat{\theta}(t))$ that satisfies the following,
    \begin{equation}
    \label{Eq:lyaponov-convergence}
        \textstyle \dot{V}_{\hat{\theta}}(\tau)  \leq 0, \quad
        \text{for all }  \; \sigma(\tau) < V_{\hat{\theta}}(\tau),\quad \begin{bmatrix}
            \widehat{q}(\tau) \\
            \widehat{p}(\tau)
        \end{bmatrix}^{\intercal} \widehat{M}(\tau)\begin{bmatrix}
            \widehat{q}(\tau) \\
            \widehat{p}(\tau)
        \end{bmatrix} \geq 0
    \end{equation}
    where $\widehat{q}_i(t) \coloneqq G_i \tilde{K}(t)\hat{\theta}(t)$, and $\widehat{p}_i(t) \coloneqq \tilde{\phi}_i(\tilde{q}_i(t)) = \phi_i(\tilde{q}_i(t)+q_i(t))-\phi_i(q_i(t))$. Expanding this condition and following the same procedure as in the proof of \Cref{Thm:RRS} yields the final result.\footnote{This optimization problem contains Differential Matrix Inequalities. Refer to \cite{Reynolds2021FunnelProblem} for details of a numeric solution for a similar problem.}

\end{proof}

We conclude the section by providing the condition on $\delta t$ such that the \ac{DG-RAN} algorithm guarantees the invariance of $\mathcal{E}_{\bar{Q}}$.

\begin{theorem}
      \label{thm: convergence}
      Consider a rapidly-regularizable nonlinear dynamic system of the form \Cref{Eq:first-order-taylor}. If the length of the measurement intervals satisfies $\delta t \leq \ell_{\theta}$, where $\ell_{\theta}$ is defined in \Cref{Eq:dt-bound}, 
      then the \ac{DG-RAN} algorithm guarantees the invariance of the set described by $\bar{Q}(t)$. Meaning  $\forall \, t \in [t_0,t_f]: \; \eta(t) \in \mathcal{E}_{\bar{Q}_t}$.
      \begin{equation}
          \label{Eq:dt-bound}
          \ell_{\theta} = \frac{\Pi_{\mathcal{R}\left(L_{k_1}^{n_t-1}\right)} {\theta}(t) - \Pi_{\mathcal{R}\left(L_{k_1}^{n_t-2}\right)} {\theta}(t)}{\Pi_{\mathcal{R}\left(L_{k_1}^{n_t-1}\right)^{\perp}} \, \dot{\theta}(t) }
      \end{equation}
\end{theorem}
\begin{proof}
    Denote $\bar{Q}(t)$ and $\mathbf{K}(t) = \bar{Q}^{-1}(t) \bar{Y}(t)$ as the solution variables to Problem \ref{Prob:funnel_synthesis}, for $t \in [t_0, t_f]$. Consider the Lyapunov function of the form
    \begin{equation}
        V_c(t) =  V(t) + \frac{1}{2}\widetilde{\theta}(t)^\intercal \Gamma^{-1} \widetilde{\theta}(t)
    \end{equation}
    where $V(t) = \eta(t)^\intercal \bar{Q}(t)^{-1} \eta(t)$ is the Lyapunov function for system \Cref{Eq:first-order-taylor} with input $\u(t) - \bar{\u}(t) =  \mathbf{K}(t)\eta(t) + \tK(t) \hat{\theta}(t)$, and $\tilde{\theta}(t)$ is defined in \Cref{Eq:sub-tilde-theta}. We need to prove there exist $\epsilon(t) \in \mathbb{R}$ such that $\dot{V}_c(t) \leq 0$ for any $\eta \in \mathcal{E}_{\bar{Q}(t)}$, where $\epsilon(t) \leq V_c(t)$. We have
    \begin{equation}
    \label{Eq:dotV_c}
        \dot{V}_c = \frac{\partial V}{\partial t} + \frac{\partial V}{\partial \eta} \dot{\eta} + \frac{\partial V}{\partial \hat{\theta}} \frac{\partial \widetilde{\theta}}{\partial t} + \widetilde{\theta}^\intercal \Gamma^{-1} \frac{\partial \widetilde{\theta}}{\partial t} = \underbrace{\dot{\eta}^\intercal \bar{Q}^{-1} \eta + \eta^\intercal \bar{Q}^{-1}  \dot{\eta} -\eta^\intercal \left(\bar{Q}^{-1} \dot{\bar{Q}} \bar{Q}^{-1} \right) \eta }_{ \dot{V}_{\hat{\theta}}(t)}+ \left(2\eta^\intercal \bar{Q}^{-1} F  + \frac{\partial V}{\partial \hat{\theta}}  + \widetilde{\theta}^\intercal \Gamma^{-1}\right) \frac{\partial \widetilde{\theta}}{\partial t} 
    \end{equation}
    where $\dot{\eta}(t) = \left( A(t)  + B(t) \mathbf{K}(t)\right)\eta(t)  + g\left(\x(t), \u(t)\right) + B(t) \mathbf{\tilde{K}}(t) \dot{\hat{\theta}}(t) + F\dot{\theta}(t)$. 

    From \Cref{thm:tilde_K}, for all time $\tau \in [t,t_f]$, there exists a data-driven feedback controller of the form $z(\tau, \mathcal{D}_t) = \tK(\tau) \hat\theta(\tau)$ such that
    $\textstyle \dot{V}_{\hat{\theta}}(\tau)  \leq 0$ for all $\sigma(\tau) < V_{\hat{\theta}}(\tau)$. Furthermore, for $t\in [n_t \delta t, (n_t+1)\delta t)$, we have
    \begin{subequations}
        \begin{equation}
        \label{Eq:partial_V_theta}
            \frac{\partial V(t)}{\partial \hat{\theta}}  = \frac{\partial}{\partial \hat{\theta}} \left(\int_{n_t \delta t}^{t} V(\tau) d\tau\right) = \int_{n_t \delta t}^{t} \eta(\tau)^{\top} \bar{Q}(\tau)^{-1}\left[B(\tau) \mathbf{\tilde{K}}(\tau)+\left(\frac{\partial g}{\partial u}\right) \mathbf{\tilde{K}}(\tau)\right] d \tau , 
        \end{equation}    
    and 
    \begin{equation}
        \label{Eq:partial_theta}
        \frac{\partial {\widetilde{\theta}}(t)}{\partial t} = \dot{\theta}(t) - \dot{\hat{\theta}}(t) = \dot{\theta}(t) - \Pi_{\mathcal{R}\left(L_{k_1}^{n_t-1}\right)} \dot{\theta}(t) -  \frac{1}{\delta t}\left( \Pi_{\mathcal{R}\left(L_{k_1}^{n_t-1}\right)} {\theta}(t) - \Pi_{\mathcal{R}\left(L_{k_1}^{n_t-2}\right)} {\theta}(t)\right)
    \end{equation}
    \end{subequations}
    where $\dot{\theta}(t) - \Pi_{\mathcal{R}\left(L_{k_1}^{n_t-1}\right)} \dot{\theta}(t) = \Pi_{\mathcal{R}\left(L_{k_1}^{n_t-1}\right)^{\perp}} \, \dot{\theta}(t)$. 
    For $t \in [0, t_f]$, let $\epsilon(t) = \sigma(\tau)$ where $\tau = t \in[t,t_f]$. By substituting \Cref{Eq:partial_V_theta} and \Cref{Eq:partial_theta} in \Cref{Eq:dotV_c}, the proof is complete. 
\end{proof}

\section{Numerical Example}
\label{sec: experiment}

This section provides a numerical example of the \ac{DG-RAN} algorithm applied to a 6-DOF powered descent guidance problem in the presence of time-varying atmospheric effects with an unknown dynamics, such as those encountered in planetary landing. This problem has nonlinear and highly unstable dynamics, state and control constraints, and large state and control dimensions that collectively make the problem challenging. In \cite{Reynolds2021FunnelProblem}, the performance of the funnel synthesis approach for the 6-DOF powered descent guidance problem in the absence of uncertainty was investigated. In the scenarios studied in \cite{Reynolds2021FunnelProblem} $\gamma-$iteration method was introduced to solve Problem \ref{Prob:funnel_synthesis}. 

\begin{table}[tb]
\caption{\label{tab:table1} The parameters for the 6-DOF powered descent case study.}
\centering
\begin{tabular}{lcccccc}
\hline
$Parameter$& & Value & & Parameter& & Value \\\hline
$J$& & $\textbf{diag}\{13600, 13600, 19150\} \, \mathrm{kg \, m}^2$  & & $\alpha_{\dot{m}}$ && $4.5324 \times 10^{-4} \mathrm{~s} / \mathrm{m}$  \\
$r_{F, \mathcal{B}}$ && $(0,0,-0.25) \mathrm{m}$ && $g_{\mathcal{I}}$ && $(0,0,-1.62) \mathrm{m} / \mathrm{s}^2$ \\
$F_{\min }$ && $5400 \mathrm{~N}$ & & $F_{\max }$ && $24750 \mathrm{~N}$ \\
$\delta_{\max }$ && $25 \mathrm{~deg}$ && $T_{\max }$ && $150 \mathrm{Nm}$ \\
$n_M$ && $ 25$ & & $N_s$ && $100$ $n_M$ \\
$\kappa_{t o l}$ && $0.5$ && $\alpha$ && $0.1 \mathrm{~s} $\\
\hline
\end{tabular}
\end{table}

The control input for the system under consideration involves a thrust vector generated by a single main engine and a reaction control system (RCS). The maneuver being performed is assumed to take place in close proximity to the landing site and over a relatively short duration, allowing for the approximation of gravity as a constant vector. Additionally, planetary rotation is neglected, and atmospheric effects are modeled as an unknown linear time-invariant (LTI) system described by \Cref{Eq:dynamic-theta}. It is assumed that the inertia matrix, center of mass, and center of pressure of the vehicle remain constant throughout the maneuver. 

To describe the vehicle's orientation, a surface-fixed landing frame denoted as $\mathcal{F}_\mathcal{I}$ is established. This frame has its origin at the intended landing site and is defined by the orthonormal vectors $\{x_\mathcal{I}, y_\mathcal{I}, z_\mathcal{I}\}$. The vectors $x_\mathcal{I}$, $y_\mathcal{I}$, and $z_\mathcal{I}$ represent the downrange, crossrange, and local up directions, respectively. Similarly, a body frame denoted as $\mathcal{F}_\mathcal{B}$ is defined with its origin at the vehicle's center of mass. The body frame is constructed using the orthonormal vectors $\{x_\mathcal{B}, y_\mathcal{B}, z_\mathcal{B}\}$, with $z_\mathcal{B}$ chosen to align with the vehicle's vertical axis. 

Due to the requirement of additive difference variables in \Cref{Eq:difference-variables-odd}, to represent the difference between the nominal state trajectory $\bar{\x}$ and the actual state trajectory $\x$, a 3-2-1 Euler angle sequence is employed to parameterize the orientation of $\mathcal{F}_{\mathcal{B}}$ with respect to $\mathcal{F}_{\mathcal{I}}$.

The state, control, and nonlinear dynamics for this problem are taken to be
\begin{equation}
    \x=\left[\begin{array}{c}
    m \\
    r_{\mathcal{I}} \\
    v_{\mathcal{I}} \\
    \Theta \\
    \omega_{\mathcal{B}}
    \end{array}\right], \quad \u=\left[\begin{array}{c}
    F_{\mathcal{B}} \\
    \tau_{\mathcal{B}}
    \end{array}\right], \quad \dot{\x}=f(\x, \u)=\left[\begin{array}{c}
    -\alpha_{\dot{m}}\left\|F_{\mathcal{B}}\right\|_2 \\
    v_{\mathcal{I}} \\
    \frac{1}{m}\left(C_{\mathcal{I} \mathcal{B}}(\Theta) F_{\mathcal{B}}\right)+g_{\mathcal{I}} \\
    T(\Theta) \omega_{\mathcal{B}} \\
    J^{-1}\left(\tau_{\mathcal{B}}+r_{u, \mathcal{B}}^{\times} F_{\mathcal{B}}-\omega_{\mathcal{B}}^{\times} J \omega_{\mathcal{B}}\right)
    \end{array}\right],
\end{equation}
where $r_{\mathcal{I}} \in \mathbb{R}^3$ is the inertial position vector, $v_{\mathcal{I}} \in \mathbb{R}^3$ is the inertial velocity vector, $\Theta=(\varphi, \theta, \psi) \in \mathbb{R}^3$ represents the Euler angles, and $\omega_{\mathcal{B}} \in \mathbb{R}^3$ is the angular velocity in the body frame. The controlled inputs are the engine thrust vector $F_{\mathcal{B}}$ and an RCS torque vector $\tau_{\mathcal{B}}$. Note that $n_x=13$ and $n_u=6$ for this problem. The true dynamics of the unknown disturbance is given by $\dot{\theta}(t) = \tilde{S} \theta(t)$ where, 
\begin{equation}
    \label{Eq:uncertainty dynamic}
    \begin{array}{cc}
         \tilde{S} &= \text{diag}\{0.001, \, -0.12, \, 0.03, \, 0.05, \, 0.07, \, 0.2, \, 0.0005\} \,  \\
         \\
        & \theta_0 = {\small 10^{-2}\begin{bmatrix}
         0.2 & 0.1 & 0.1 & 3 & 1 & 2 & 3
    \end{bmatrix}}^\intercal \normalsize 
    \end{array}
     \, , \qquad F ={\small \begin{bmatrix}
        0_{2\times 2} & 0_{2\times 2}  & 0_{2\times 3}\\
        I_{2} & 0_{2\times 2}  & 0_{2\times 3}\\
        0_{1\times 2} & 0_{1\times 2}  & 0_{1\times 3}\\
        0_{2\times 2} & I_2 & 0_{2\times 3}\\
        0_{4\times 2} & 0_{4\times 2} & 0_{4\times 3} \\
        0_{3\times 2} & 0_{3\times 2} & I_3 
    \end{bmatrix}} \normalsize
\end{equation}
and $F$ is the known basis function capturing the effects of the disturbance on the system dynamics. The disturbances in the system primarily affect the linear velocity and acceleration along the $x$-axis and $y$-axis, as well as the angular velocity in all three directions. These disturbances introduce variations and uncertainties in the motion of the system, making the control problem more challenging. By addressing the impact of these disturbances, our proposed approach aims to regulate the system states and mitigate their effects, ultimately improving the overall performance and stability of the system.

The nominal trajectory was determined using the PTR algorithm, as described in \cite{reynolds2020dual}, with specific boundary conditions and a final time of $t_f=29.7$ s. The initial conditions at $t_0$ were set as follows: $\bar{m}=3250$ kg, $\bar{r}_{\mathcal{I}}=(250,0,433)$ m, $\bar{v}_{\mathcal{I}}=(-35.7,0,-11.8)$ m/s, $\Theta=(0,59.8,0)$ degrees, and $\omega_{\mathcal{B}}=0_{3 \times 1}$. The final conditions at $t_f$ were set as: $\bar{m}=3130.3$ kg, $\bar{r}_{\mathcal{I}}=(0,0,30)$ m, $\bar{v}_{\mathcal{I}}=(0,0,-1)$ m/s, $\Theta=0_{3 \times 1}$, and $\omega_{\mathcal{B}}=0_{3 \times 1}$. 
%
% In this example, we impose two nonlinear control constraints, and the nominal trajectory is computed while considering both upper and lower thrust bounds: $F_{\min } \leq\left\|F_{\mathcal{B}}\right\|_2 \leq F_{\max }$, as well as torque bounds: $-T_{\max } \leq\left\|\tau_{\mathcal{B}}\right\|_{\infty} \leq T_{\max }$. Here, $F_{\min }, F_{\max }, T_{\max } \in \mathbb{R}_{++}$ represent specific values for the bounds. Additionally, a gimbal angle constraint is incorporated: $\left\|F_{\mathcal{B}}\right\|_2 \leq \sec \delta_{\max } z_{\mathcal{B}}^{\top} F_{\mathcal{B}}$.
%
The constraint set for the state space $\mathcal{X}$ is then taken to be

\begin{align}
    \mathcal{X} & =\left\{x \mid x_{l b} \leq x \leq x_{u b}, \delta x_{l b} \leq x-\bar{x} \leq \delta x_{u b},\|x\|_2 \leq \infty\right\} \, , \quad \mathcal{X}_f =\mathcal{E}_{Q_{\max , f}} \, , 
\end{align}
where $\mathcal{X}$ enforces both an absolute bound on the state vector as well as a bound on the deviation from the nominal trajectory.
\begin{align*}
\tiny
\textstyle x_{l b} & =-(-2100,150,150,0,40,40,30, \pi, \pi/2, \pi, 0.5,0.5,0.5) \\
\textstyle x_{u b} & =+(3737.7,350,300,500,30,30,5, \pi, \pi/2, \pi, 0.5,0.5,0.5) \\
\textstyle \delta x_{l b} & =-(\infty, 100,100,100, \infty, \infty, \infty, 2/9 \pi, 2/9 \pi, 2/9 \pi, 2/9 \pi, 2/9 \pi, 2/9 \pi), \\
\textstyle \delta x_{u b} & =+(\infty, 100,100,100, \infty, \infty, \infty, 2/9 \pi, 2/9 \pi, 2/9 \pi, 2/9 \pi, 2/9 \pi, 2/9 \pi),
\normalsize
\end{align*}
The terminal constraints specified by $\mathcal{X}_{\mathcal{F}}$ guarantee that every trajectory remains within a maximum deviation of $0.5$ m in position, $0.25$ m/s in velocity, and $3$ degrees in attitude and $3$ degrees/s in angular rate, from the nominal path at the final time. The additional information for this particular problem can be found in \Cref{tab:table1}, which is loosely based on a lander similar to the Apollo class. The matrices $A$ and $B$ represent the partial derivatives of $f$ along the nominal trajectory, while the parameters $H$, $G$, and $E$ are derived from a set of six $(n_p = 6)$ nonlinear channels for $\delta p$ and one $(n_{\tilde{p}} = 1)$ nonlinear channel  for $\delta \tilde{p}$.
\begin{align*}
    \small H=\left[\begin{array}{ccccc}
    I_1 & 0_{1 \times 3} & 0_{1 \times 3} & 0_{1 \times 3} & 0_{1 \times 3} \\
    0_{3 \times 1} & 0_{3 \times 3} & 0_{3 \times 3} & I_3 & 0_{3 \times 3} \\
    0_{3 \times 1} & 0_{3 \times 3} & 0_{3 \times 3} & I_3 & 0_{3 \times 3} \\
    0_{3 \times 1} & 0_{3 \times 3} & 0_{3 \times 3} & 0_{3 \times 3} & I_3 \\
    0_{3 \times 1} & 0_{3 \times 3} & 0_{3 \times 3} & 0_{3 \times 3} & I_3 \\
    0_{3 \times 3} & 0_{3 \times 3} & 0_{3 \times 3} & 0_{3 \times 3} & 0_{3 \times 3}
    \end{array}\right], \; G=\left[\begin{array}{cc}
    0_{1 \times 3} & 0_{1 \times 3} \\
    0_{3 \times 3} & 0_{3 \times 3} \\
    0_{3 \times 3} & 0_{3 \times 3} \\
    0_{3 \times 3} & 0_{3 \times 3} \\
    0_{3 \times 3} & 0_{3 \times 3} \\
    I_3 & 0_{3 \times 3}
    \end{array}\right], \; E=\left[\begin{array}{cccccc}
0_{1 \times 3} & 0_{1 \times 3} & 0_{1 \times 3} & 0_{1 \times 3} & 0_{1 \times 3} & I_1 \\
0_{3 \times 3} & 0_{3 \times 3} & 0_{3 \times 3} & 0_{3 \times 3} & 0_{3 \times 3} & 0_{1 \times 1} \\
I_3 & I_3 & 0_{3 \times 3} & 0_{3 \times 3} & 0_{3 \times 3} & 0_{1 \times 1} \\
0_{3 \times 3} & 0_{3 \times 3} & I_3 & I_3 & 0_{3 \times 3} & 0_{1 \times 1} \\
0_{3 \times 3} & 0_{3 \times 3} & 0_{3 \times 3} & 0_{3 \times 3} & I_3 & 0_{1 \times 1}
\end{array}\right]
\normalsize
\end{align*}

\begin{figure}[t!]
\centering
\subfloat[\label{fig:x-v-offline}position and linear velocity in state space $\mathcal{E}_Q$]{%
  \includegraphics[clip,width=.75\columnwidth]{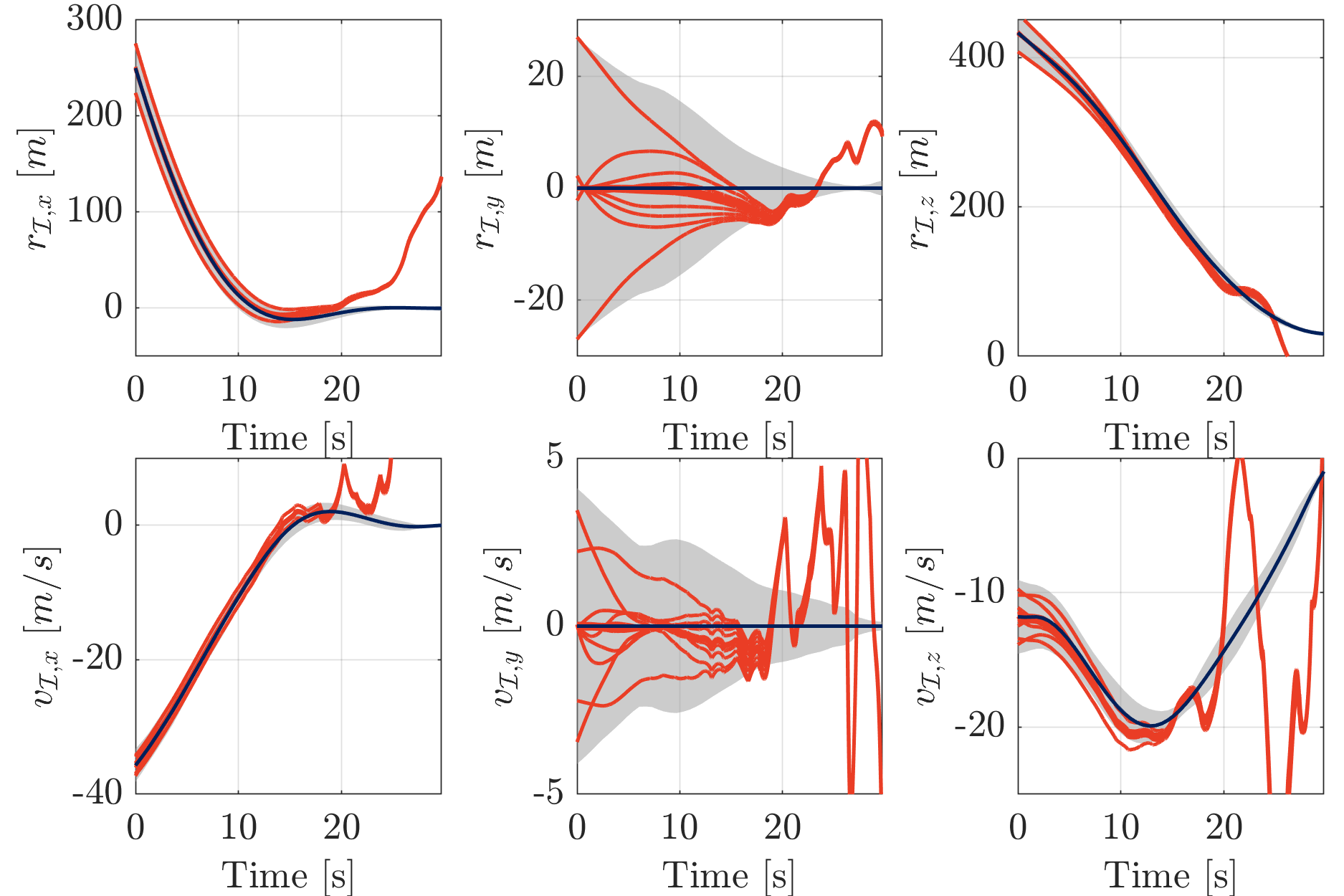}%
}

\subfloat[\label{fig:t-w-offline}rotation angles and angular velocity in state space $\mathcal{E}_Q$]{%
  \includegraphics[clip,width=.75\columnwidth]{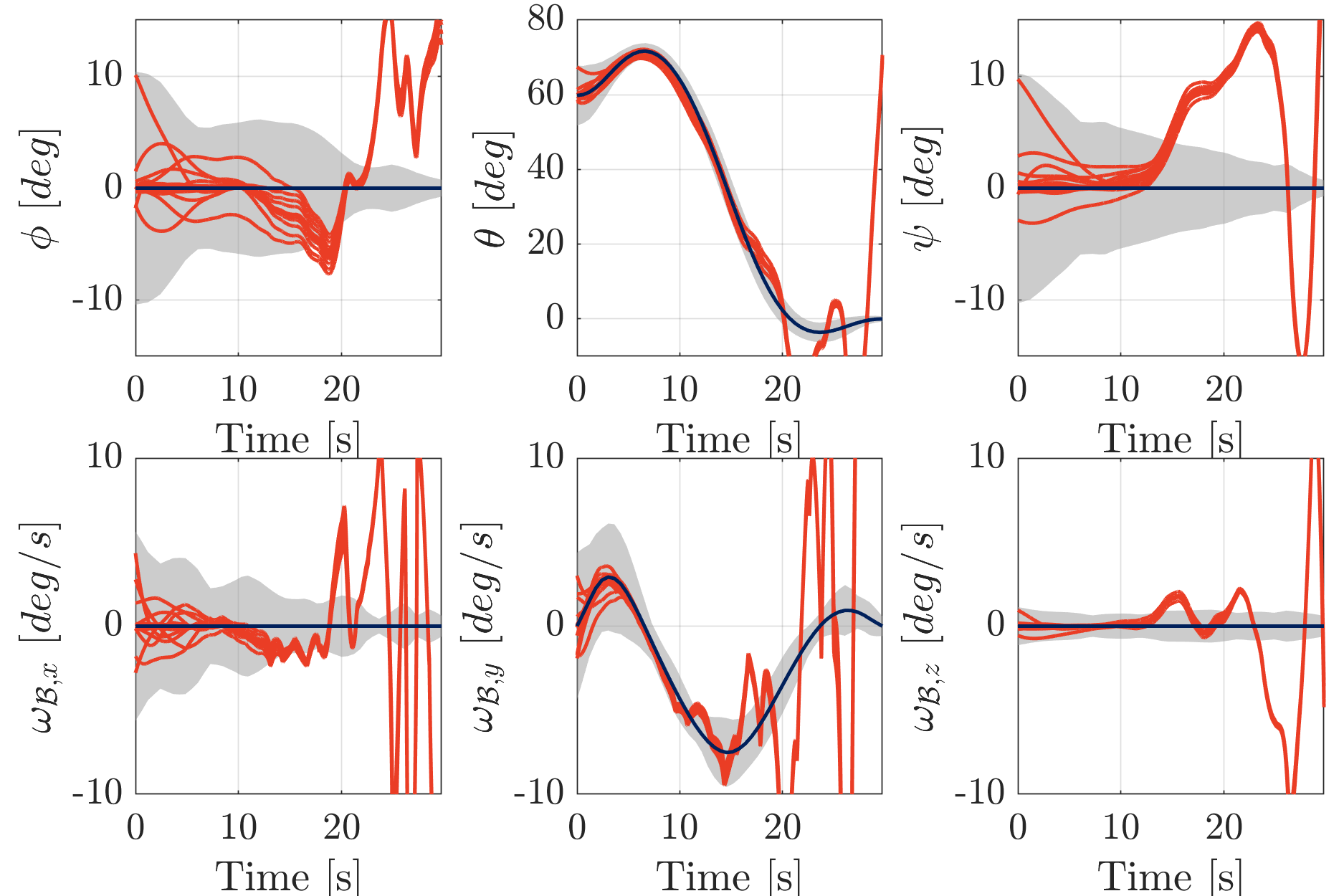}%
}

\caption{The quadratic funnel computed by the $\gamma$-iteration for the 6-DOF power descent problem in the presence of unmodeled time-varying disturbances. The initial condition of each test case was randomly sampled from the funnel entry, and the system uses only the offline robust controller introduced in Problem \ref{Prob:funnel_synthesis}.}
\label{fig:offline}
\end{figure}

\Cref{fig:offline} shows the trajectories of the system in \Cref{Eq:first-order-taylor} generated using only the offline robust controller. The ellipsoid $\mathcal{E}_{\bar{Q}}$ is projected onto each state dimension (mass is omitted) and depicted as the shaded grey area. The red trajectories correspond to test cases for which an initial condition was randomly (uniformly) selected from the ellipsoid $\mathcal{E}_{\bar{Q}\left(t_0\right)}$, and the nominal control and correction law $(\xi(t) = \mathbf{K}(t)\eta(t))$ were used to integrate the equations of motion numerically.
\begin{subequations}
    \begin{align}
     u(t) &=\bar{u}(t)+K(t)(x(t)-\bar{x}(t)) \label{subEq:u-offline} \\
     x(t) & =x\left(t_0\right)+\int_{t_0}^t \left(f(x(\tau), u(\tau)) + F \tilde{S}\theta(\tau) \right) \mathrm{d} \tau
    \end{align}
\end{subequations}
In the presence of unmodeled disturbance, the offline controller fails to keep the states inside the quadratic state funnel, which results in undesirable behavior of the system.

\begin{figure}[t!]
\centering
\subfloat[\label{fig:x-v-online}position and linear velocity in state space $\mathcal{E}_Q$]{%
  \includegraphics[clip,width=.75\columnwidth]{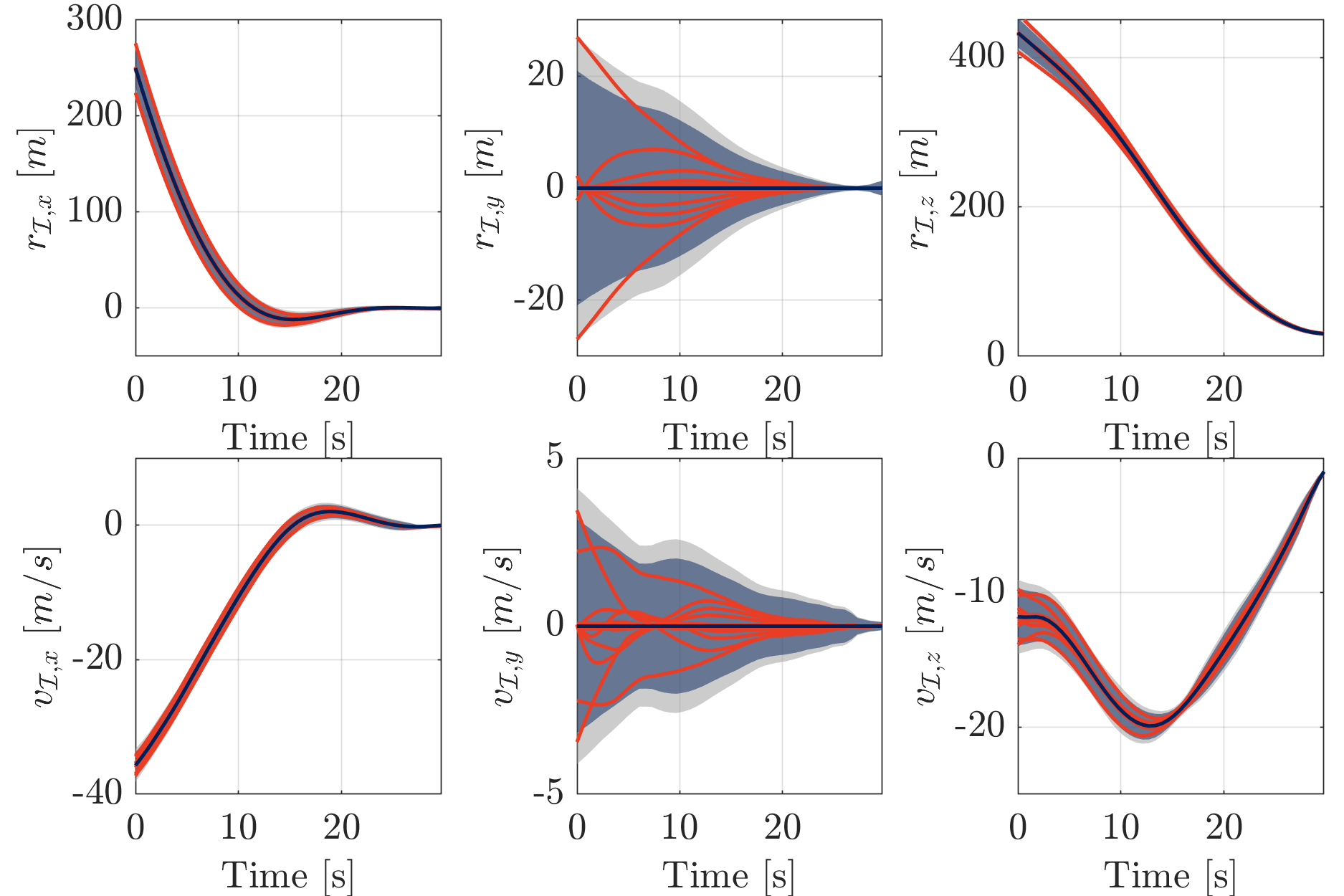}%
}

\subfloat[\label{fig:t-w-online}rotation angles and angular velocity in state space $\mathcal{E}_Q$]{%
  \includegraphics[clip,width=.75\columnwidth]{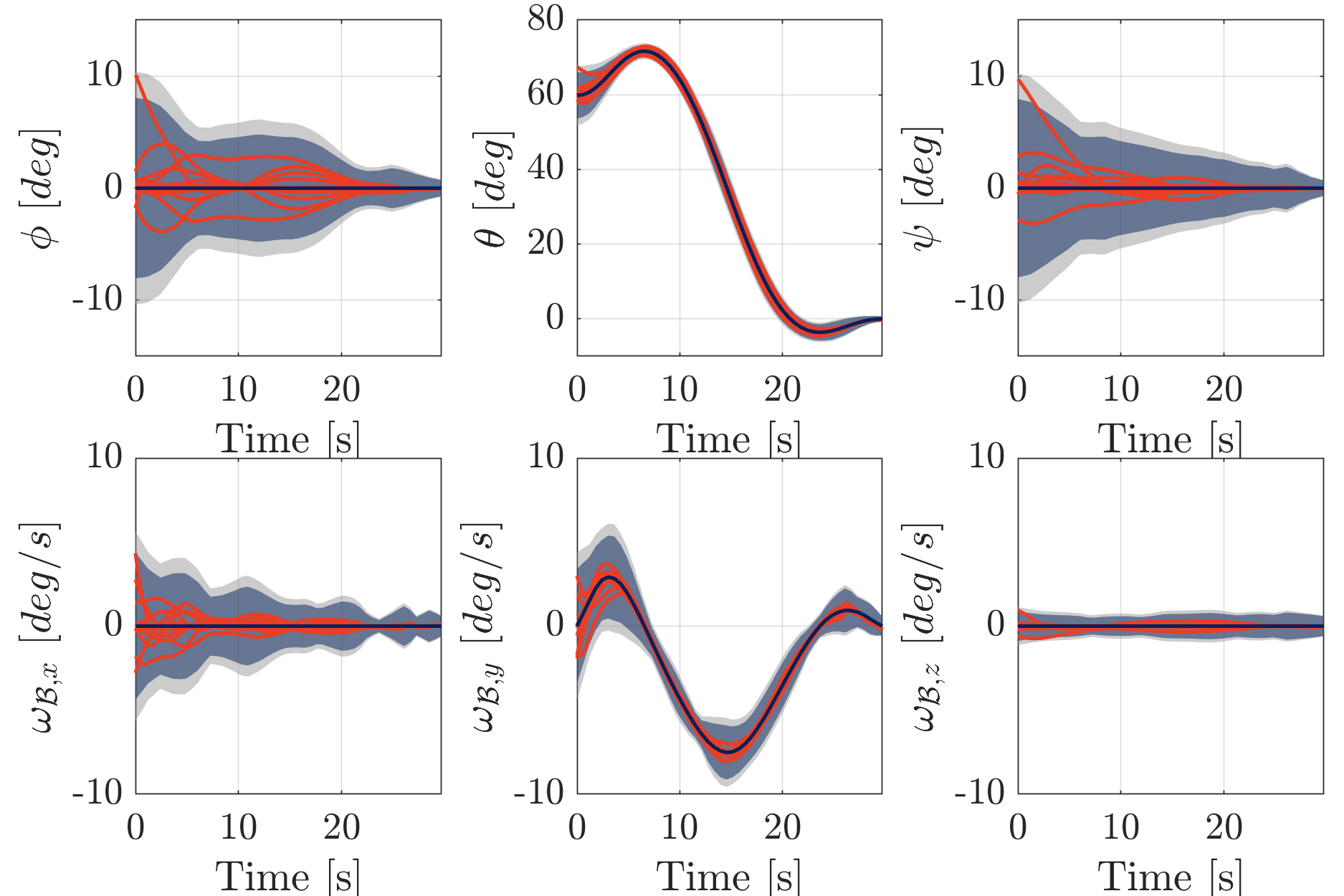}%
}

\caption{State trajectories generated using the online DG-RAN algorithm for the 6-DOF power descent problem in the presence of unmodeled time-varying disturbances. The initial condition for this test case was randomly sampled from the funnel entry, and the proposed method was implemented online for each test case. The shaded gray area represents $\mathcal{E}_{\bar{Q}(t)}$ projected onto each state, and the shaded dark blue area represents the projection of the inner invariant funnel $\mathcal{E}(t)$ onto each state.}
\label{fig:online}
\end{figure}

On the contrary, \Cref{fig:online} showcases the performance of the \ac{DG-RAN} algorithm. The ellipsoid $\mathcal{E}_{\bar{Q}}$ is projected onto each state dimension and depicted as the shaded grey area, where the shaded blue area demonstrates the projection of ellipsoid $\mathcal{E}$, \Cref{Eq:invariant-ellips} onto each state dimension, starting at time $t=n_t \delta t$ where $n_t=10$ and $\delta t = 0.5 \mathrm{~s}$. Note that \ac{DG-RAN} is an online algorithm where the proposed method was implemented uniquely for each trajectory to estimate the model of the disturbance dynamic and update $z(t,\mathcal{D}_t) = \tK(t) \hat{\theta}(t)$ at each measurable time $t= n_t \delta t$.
\begin{align}
 u(t) &=\bar{u}(t)+K(t)(x(t)-\bar{x}(t)) + \tK(t) \hat{\theta}(t)  \label{subEq:u-online}
\end{align}
The proposed method can identify all modes of disturbance dynamic affecting the system at time $t=5.5 \mathrm{~s}$ with a data set $\mathcal{D}_t$ of size $n_t=11$.

One important observation is how the \ac{DG-RAN} algorithm identifies the disturbance modes in the order of severity. Modes of the disturbance dynamic causing the trajectory to deviate from the nominal are quickly identified, and $z(t,\mathcal{D}_t)$ is updated to reflect such deviations. In particular, the data is informative enough to identify unknown parameters through least squares denoted by $\hat{S}_{n_t}$. Therefore, one stopping criterion–which is also used here–is the point where the estimate of parameters $\hat{S}_{n_t}$ has converged.
Then, $z(t,\mathcal{D}_t)$ is implemented without updates from $t=$ onward. In contrast to using only the offline controller \Cref{subEq:u-offline} (Fig aa), it is shown that the updated controller \Cref{subEq:u-online} is robust to the effects of environmental disturbances; since we now have a more accurate estimate of the disturbance dynamic using the data generated safely by DG-RAN in the loop.

These results are quite promising - the powered descent problem is a challenging problem with nonlinear dynamics, some nonlinear and nonconvex constraints, and relatively large state and control dimensions. What these results show is that in the presence of unmodeled disturbances, we are able to generate feasible trajectories (both dynamically and with respect to the constraints that were considered) for any initial condition in a set that stretches more than $35 \mathrm{~m}$ in every position direction, $2 \mathrm{~m} / \mathrm{s}$ in each velocity direction, $16 \mathrm{deg}$ in each Euler angle, and $2 \mathrm{deg} / \mathrm{s}$ in each angular velocity direction, all by using online updates.

\section{Conclusion}
\label{sec: conclusion}

In this paper, we addressed the problem of designing data-driven feedback controllers for complex nonlinear dynamical systems in the presence of unknown time-varying disturbances with LTI dynamics. Our goal was to achieve {\em finite-time regulation} of system states while ensuring the invariance of the state space funnel. We identified the limitations of existing adaptive control approaches and proposed a novel framework to overcome these challenges.

By expanding upon the notion of ``regularizability'' introduced in \cite{talebi2021regularizability}, we extended its applicability to a class of partially unknown nonlinear systems, where the time-varying disturbance dynamics was modeled as the unknown part of the system dynamics. We introduced the concept of rapid-regularizability and provided a characterization for a linear time-varying representation of the nonlinear system with locally-bounded higher-order terms. This characterization allowed us to analyze the spectral properties of the underlying system, providing insights into its behavior.

To address the problem of online regulation for partially unknown nonlinear systems, we proposed the \ac{DG-RAN} algorithm. This algorithm offered an iterative synthesis procedure that leveraged discrete time-series data from a single trajectory. By utilizing this data, the \ac{DG-RAN} algorithm achieved effective regulation of system states while simultaneously generating informative data for identifying the dynamics of the disturbance. This online approach was advantageous in terms of control synthesis efficiency and adaptability to changing environments.

In conclusion, our work contributed to the integration of learning, optimization, and control theory by addressing the limitations of existing methods and introducing a novel framework for designing adaptive controllers. The proposed approach enabled autonomous systems to adapt quickly to new environments while ensuring safety and control-theoretic guarantees. Our research opens up avenues for future exploration in online safety-critical control applications. Further investigations can focus on refining the \ac{DG-RAN} algorithm, extending the framework to handle more complex system dynamics, and conducting real-world experiments to validate its performance.

% \section*{Appendix}
% \subsection{Approximation error for integral}
% \label{Appendix1}

\section*{Acknowledgments}

\bibliography{main}

\end{document}